\documentclass{article}
\usepackage{amsmath, amsfonts,amsthm, amssymb}
\usepackage{graphics,hyperref,url}
\usepackage[english]{babel}

\newtheorem{theorem}{Theorem}
\newtheorem{assumption}{Assumption}

\newtheorem{proposition}{Proposition}
\newtheorem{remark}{Remark}
\newtheorem{lemma}{Lemma}

\title{Forward equations for option prices in semimartingale models}

\date{Laboratoire de Probabilit\'es et
    Mod\`eles Al\'eatoires,\\
     CNRS - Universit\'e de Paris    VI\\
    and\\
      Columbia University, New York.\\
     \ \\
     First version: Sept  2009. Final revision: 2011.\\
     To appear in: {\it Finance \& Stochastics.}\footnote{We thank Bruno Dupire, Jean Jacod and seminar participants at the EURANDOM Workshop on Statistical Inference for L\'evy Processes, Workshop on PDEs and mathematical finance (KTH 2009), Research in Options (2009),
     Conference on PDEs and Mathematical finance (Rutgers 2009), Bachelier Congress (2010), Columbia University and Universit\'e Paris-Est  for helpful comments. }}
\author{Amel Bentata\  and Rama Cont}

\begin{document}

\maketitle

\begin{abstract}
We derive a forward partial integro-differential equation for
prices of call options in a model where the dynamics of the
underlying asset under the pricing measure is described by a
-possibly discontinuous- semimartingale. This result generalizes
Dupire's forward equation to a large class of non-Markovian models with jumps.
\end{abstract}
\tableofcontents
\newpage

Since the seminal work of Black, Scholes and Merton
\cite{blackscholes,merton73} partial differential equations (PDE)
have been used as a way of characterizing and efficiently computing
option prices. In the Black-Scholes-Merton model and various
extensions of this model which retain the Markov property of the
risk factors, option prices can be characterized in terms of
solutions to a backward PDE, whose variables are time (to maturity)
and the value of the underlying asset. The use of backward PDEs for
option pricing has been extended to cover options with
path-dependent and early exercise features, as well as to
multifactor models (see e.g. \cite{achdoupironneau}).
 When the underlying asset exhibit jumps, option prices can be computed by solving an analogous partial integro-differential equation (PIDE) \cite{andersen,cont05}.

A second important step was taken by  Dupire \cite{dupire93,dupire94,dupire97} %(see also Derman \& Kani \cite{dermankani94})
who showed that when the underlying asset is assumed to follow a diffusion process
$$ dS_t= S_t \sigma(t,S_t) dW_t $$ prices of call options (at a given date $t_0$) solve a {\it forward} PDE
$$
  \frac{\partial C_{t_0}}{\partial
T}(T,K) = -r(T)K\frac{\partial C_{t_0}}{\partial K}(T,K)+\frac{K^2\sigma(T,K)^2}{2}\, \frac{\partial^2 C_{t_0}}{\partial
K^2}(T,K)
$$
 on $[t_0,\infty[\times ]0,\infty[$ in the strike and maturity  variables, with
the initial condition $$\forall K>0\quad C_{t_0}(t_0,K)= (S_{t_0}-K)_+.$$
 This forward equation allows to price call options with various strikes and maturities on the same underlying asset, by solving a {\it single} partial differential equation. Dupire's forward equation also provides useful insights into the {\it inverse problem} of calibrating  diffusion models to observed call and put option prices \cite{busca02}.

Given the theoretical and computational usefulness of the forward equation, there have been various attempts to extend Dupire's forward equation to other types of options and processes, most notably to Markov processes with jumps \cite{andersen,carrhirsa,contsavescu08,jourdain07,carrgeman04}. Most of these constructions use the Markov property of the underlying process in a crucial way (see however \cite{klebaner02}).

As   noted by Dupire \cite{dupire96}, the forward PDE holds in a more general context than the backward PDE: even if the (risk-neutral) dynamics of the underlying asset is not necessarily Markovian, but described by a continuous Brownian martingale
 $$ dS_t= S_t\delta_t dW_t, $$
 then call options still verify a forward PDE where the diffusion coefficient is given by the local (or effective) volatility function $\sigma(t,S)$ given by
$$ \sigma(t,S)=\sqrt{E[\delta_t^2|S_t=S]}. $$
This method is linked to the ``Markovian projection" problem: the
construction of a Markov process which mimicks the marginal distributions of a martingale \cite{bentatacont09,gyongy86,madanyor02}. Such ``mimicking processes" provide a method to extend the Dupire equation to non-Markovian settings. %Non-Markovian mimicking processes have been also considered \cite{brunick09} but do not lead to PDEs.

   We show in this work that the forward equation for call prices holds in a more general setting, where the dynamics of the underlying asset is described by a -- possibly
    discontinuous -- semimartingale. Our parametrization of the price dynamics is general, allows for stochastic volatility and does {\it not} assume jumps to be independent or driven by a L\'evy process, although it includes these cases.
%We present two different derivations of this result. The first derivation (Section \ref{ch3.forwardpide.sec}) is based on a generalization \cite{bentatacont09} of Gy\"ongy's \cite{gyongy86} ``mimicking theorem" to  semimartingales.
 Also, our derivation  does not require ellipticity or non-degeneracy of the
diffusion coefficient. The result is thus applicable to various stochastic volatility models with jumps, pure jump models and point
process models used in equity and credit risk modeling.

Our result extends the forward equation from the original diffusion setting of Dupire \cite{dupire94} to  various examples of non-Markovian and/or discontinuous processes and implies previous derivations of forward equations  \cite{andersen,carrhirsa,carrgeman04,contsavescu08,dupire94,dupire96,jourdain07,lopatin08} as  special cases.
Section \ref{ch3.examples.sec} gives  examples of  forward PIDEs  obtained in various
settings:  time-changed L\'{e}vy processes, local L\'evy models and point processes used in portfolio default risk modeling.
In the case where the underlying risk factor follows, an It\^o process or a Markovian jump-diffusion driven by a L\'evy process, we retrieve previously known forms of the forward equation. In this case, our approach gives a rigorous derivation of these results under precise assumptions in a unified framework.
In some cases, such as index options (Sec. \ref{ch3.index.sec}) or CDO expected tranche notionals (Sec. \ref{ch3.forwardcdo.sec}), our method leads to a new,  more general form of the forward equation valid for a larger class of models than previously studied \cite{avellaneda03,contsavescu08,schonbucher05}.

The forward equation for call options is a PIDE in one (spatial) dimension, regardless of the number of factors driving the underlying asset. It may thus be used as a method for reducing the dimension of the problem.
The case of index options (Section \ref{ch3.index.sec}) in  a multivariate jump-diffusion model illustrates how the forward equation projects a high dimensional pricing problem into a one-dimensional state equation.

\section{Forward PIDEs for call options}\label{forwardpide.sec}
\subsection{General formulation of the forward equation}
Consider a (strictly positive)  semimartingale $S$ whose dynamics under the pricing measure
$\mathbb{P}$ is given by
\begin{equation}\label{ch3.stochmodel}
S_T = S_0 +\int_0^T r(t) S_{t^-} dt + \int_0^T S_{t^-}\delta_t dW_t + \int_0^T\int_{-\infty}^{+\infty} S_{t^-}(e^y - 1) \tilde{M}(dt\ dy),
\end{equation}
where  $r(t)>0$ represents a (deterministic) bounded discount rate,
$\delta_t$ the (random) volatility process and $M$ is an integer-valued
random measure with compensator
$$\mu(dt\,dy; \omega)=m(t,dy,\omega)\,dt,$$
representing jumps in the log-price, and $\tilde{M}=M-\mu$ is the compensated random measure associated to $M$
(see \cite{contankov} for further background).
Both the volatility $\delta_t$ and  $m(t,dy)$, which represents the intensity of jumps of size $y$ at time $t$, are allowed to be stochastic. In particular, we do {\it not} assume the jumps to be driven by a L\'evy process or a process with independent increments.
The specification \eqref{ch3.stochmodel} thus includes most stochastic volatility models with jumps.

We assume the following  conditions:
\begin{assumption}[Full support]\label{ch3.fullsupport.cond}
  $\forall t\geq 0$, ${\rm supp}(S_{t})=[0,\infty[$.
\end{assumption}
\begin{assumption}[Integrability condition]\label{ch3.integrability.cond}
\begin{equation}\tag{H}\label{ch3.H}
\forall T>0, \quad
\mathbb{E}\left[\exp{\left(\frac{1}{2}\int_0^T\delta_t^2\,dt
      + \int_0^T dt \int_{\mathbb{R}} (e^y-1)^2 m(t,dy)\right)}\right]<\infty.\quad
\end{equation}
\end{assumption}
The value $C_{t_0}(T,K)$ at  $t_0$ of a call option with expiry $T>t_{0}$
and strike $K>0$ is given by
\begin{equation}\label{ch3.def.call}
C_{t_0}(T,K)=e^{-\int_{t_0}^T r(t)\,dt}E^{\mathbb{P}}[\max(S_T-K,0)|\mathcal{F}_{t_0}].
\end{equation}
As argued in Section \ref{ch3.tanaka.sec}, under Assumption \eqref{ch3.H},
the expectation in \eqref{ch3.def.call} is finite. Our main result is the following:
\begin{theorem}[Forward PIDE for call options]\label{ch3.pide.forward.prop}
%If the functions
%\begin{equation}\label{ch3.hypo2}
%\sigma(t,z)=\sqrt{\mathbb{E}\left[\delta_t^2|S_{t^-}=z\right]};\qquad
%n(t,y,z)=\mathbb{E}\left[ m(t,y)|S_{t^-}=z\right].
%\end{equation}
%are continuous in $(t,z)$ on $[0,T]\times [0,\infty[$.

%\begin{eqnarray}
%  \forall T>0 \quad \mathbb{E}\Big[\exp\Big(&&\frac{1}{2}\int_0^T\delta_t^2\,dt\\[0.1cm]
%      &+& \int_0^T \int_{y\leq 1} 1\wedge y^2 \, m(t,y) dy  \quad\quad(A_{a})\\[0.1cm]
%      &+& \int_0^T \int_{ y>1} e^{2y} \, m(t,y) dy\Big)\Big] <\infty\quad a.s.
     %&\mathbb{E}\left[\int_{ y>1}e^y\, m(t,dy,\omega)\right]<\infty \quad a.s. \quad(A_{3a})
    % & \forall T\quad \mathbb{E}\left[\int_0^T \int_{ y>1} e^{y} m(t,y,\omega) dy\right] <\infty\quad a.s.(A'_3)
 %\end{eqnarray}
Let $\psi_t$  be the exponential double tail of the compensator $m(t,dy)$
\begin{equation}
  \psi_{t}(z)=
  \begin{cases}
&\int_{-\infty}^z dx\  e^x \int_{-\infty}^x m(t,du), \quad z<0\ ;\\
&\int_{z}^{+\infty} dx\  e^x \int_x^{\infty} m(t,du),\quad z>0\\
\end{cases}
\end{equation}
and let $\sigma:[t_0,T]\times \mathbb{R}^+-\{0\}\mapsto \mathbb{R}^+$, $\chi:[t_0,T]\times \mathbb{R}^+-\{0\}\mapsto \mathbb{R}^+$ be measurable functions such that for all $t\in[t_0,T]$
%and define, for $t> t_0, z>0$,
\begin{equation}\label{ch3.new.para}
\begin{cases}
 \sigma(t,S_{t-})&=\sqrt{\mathbb{E}\left[\delta_t^2|S_{t^-}\right]};\\
 \chi_{t,S_{t-}}(z)&=\mathbb{E}\left[\psi_t\left(z\right)|S_{t-}\right]\qquad a.s.
\end{cases}
\end{equation}
Under assumption \eqref{ch3.H}, the call
option price $(T,K)\mapsto C_{t_0}(T,K)$, as a function of maturity and
 strike, is a solution (in the sense of distributions) of the partial
integro-differential equation
\begin{eqnarray}
  \frac{\partial C_{t_0}}{\partial
T}(T,K) = -r(T)K\frac{\partial C_{t_0}}{\partial K}(T,K)+\frac{K^2\sigma(T,K)^2}{2}\, \frac{\partial^2 C_{t_0}}{\partial
K^2}(T,K)\nonumber\\
 +\int_{0}^{+\infty} y\frac{\partial^2 C_{t_0}}{\partial K^2}(T,dy)\,\chi_{T,y}\left(\ln{\left(\frac{K}{y}\right)}\right)
\label{ch3.pide.forward.eq}\end{eqnarray}
on $[t_0,\infty[\times ]0,\infty[$ with
the initial condition:
$$\forall K>0,\quad C_{t_0}(t_0,K)= (S_{t_0}-K)_+.$$
%Furthermore if either \\$\delta_t\geq c>0\ $a.s. or the compensator $m$ verifies  \eqref{ch3.stable1.eq}-\eqref{ch3.stable2.eq}, then $C\in
%C^{1,2}$ is a classical solution of \eqref{ch3.pide.forward.eq}.
\end{theorem}
\begin{remark} Recall that $f: [t_0,\infty[\times ]0,\infty[\mapsto \mathbb{R}$
is a solution of \eqref{ch3.pide.forward.eq} in the sense of distributions on $[t_0,\infty[\times]0,\infty[$ if
for any test function $ \varphi\in
C_0^\infty([t_0,\infty[\times ]0,\infty[,\mathbb{R})$ and for any $T\geq t_0$,
\begin{eqnarray*}
&&\int_{t_0}^Tdt\,\int_0^\infty  dK \varphi(t,K)\,\Big[-\frac{\partial f}{\partial t}
-r(t)K\frac{\partial f}{\partial K}+ \frac{K^2\sigma(t,K)^2}{2}\,
\frac{\partial^2 f}{\partial K^2} \\
&&\quad\quad\quad \quad\quad\quad\quad\quad\quad+\int_{0}^{+\infty}
y\frac{\partial^2 f}{\partial
K^2}(t,dy)\,\chi_{t,y}\left(\ln{\big(\frac{K}{y}\big)}\right)\Big]=0,
\end{eqnarray*}
where $C_0^\infty([t_0,\infty[\times]0,\infty[,\mathbb{R})$ is the set of infinitely
differentiable functions with compact support in $[t_0,\infty[\times]0,\infty[$. This
notion of generalized solution allows to separate the discussion of
existence of solutions from the discussion of their regularity
(which may be delicate, see \cite{cont05}).
\end{remark}
\begin{remark}\label{ch3.novikov.rem} The discounted asset price
 $$\hat{S}_T= e^{-\int_0^T r(t)dt}\,S_T,$$ is the stochastic exponential of the martingale $U$ defined by
$$U_T=\int_0^T \delta_t\,dW_t+\int_0^T\int (e^y-1)\tilde{M}(dt\,dy).$$
Under assumption \eqref{ch3.H}, we have
$$
\forall T>0,\quad\mathbb{E}\left[\exp{\left(\frac{1}{2}\langle U,
U\rangle_T^d + \langle U,U\rangle_T^c\right)}\right]<\infty,
$$
where $\langle U,U\rangle^c$ and $\langle U,U\rangle^d$ denote the continuous and purely discontinuous parts of $[U,U]$. \cite[Theorem 9]{protter08} implies that $(\hat{S}_T)$ is a $\mathbb{P}$-martingale.
\end{remark}
The form of the integral term in \eqref{ch3.pide.forward.eq} may seem
different from the integral term appearing in backward PIDEs
\cite{cont05,schwab09}. The following lemma expresses
$\chi_{T,y}(z)$ in a more familiar form in  terms of call payoffs:
\begin{lemma}\label{ch3.calcul.de.B}
Let $n(t,dz,y,\omega)\,dt$
be a random measure on $[0,T]\times\mathbb{R}\times\mathbb{R}^+$ verifying
$$ \forall t\in[0,T],\quad \int_{-\infty}^\infty (e^z\wedge |z|^2) n(t,dz,y,\omega) < \infty\qquad {\rm a.s.}$$
 %representing the jumps of size $y$ at time $t$, depending on the space variable $z$ and
Then the exponential double tail  $\chi_{t,y}(z)$ of $n$, defined as
\begin{equation}
  \chi_{t,y}(z)=
  \begin{cases}
&\int_{-\infty}^z dx\  e^x \int_{-\infty}^x n(t,du,y), \quad z<0\ ;\\
&\int_{z}^{+\infty} dx\  e^x \int_x^{\infty} n(t,du,y),\quad z>0\\
\end{cases}
\end{equation}
verifies
$$
\int_{\mathbb{R}}[(ye^z-K)^+-e^z(y-K)^+-K(e^z-1)1_{ \{y>K\}}]n(t,dz,y)=y\,\chi_{t,y}\left(\ln{\left(\frac{K}{y}\right)}\right).
$$
\end{lemma}
\begin{proof} Let $K,T >0$. Then
  \begin{eqnarray*}
&&\int_{\mathbb{R}}[(ye^z-K)^+-e^z(y-K)^+-K(e^z-1)1_{ \{y>K\}}]n(t,dz,y) \nonumber\\
&=&\int_{\mathbb{R}}[(ye^z-K)1_{ \{z>\ln{(\frac{K}{y})}\}}-e^z(y-K)1_{\{ y>K\}}-K(e^z-1)1_{ \{y>K\}}]n(t,dz,y)\nonumber\\
&=&\int_{\mathbb{R}}[(ye^z-K)1_{ \{z>\ln{(\frac{K}{y})}\}}+(K-ye^z)1_{\{ y>K\}}]n(t,dz,y).\nonumber
  \end{eqnarray*}
\begin{itemize}
\item If $K\geq y$, then
\begin{eqnarray*}
&&\int_{\mathbb{R}} 1_{ \{K\geq y\}}[(ye^z-K)1_{\{ z>\ln{(\frac{K}{y})}\}}+(K-ye^z)1_{ \{y>K\}}]n(t,dz,y)\nonumber\\
&=&\int_{\ln{(\frac{K}{y})}}^{+\infty} y(e^z-e^{\ln{(\frac{K}{y})}})\,n(t,dz,y).
\end{eqnarray*}
\item If $K<y$, then
\begin{eqnarray*}
&&\int_{\mathbb{R}} 1_{ \{K< y\}}[(ye^z-K)1_{ \{z>\ln{(\frac{K}{y})}\}}+(K-ye^z)1_{ \{y>K \}}]n(t,dz,y)\nonumber\\
&=&\int_{\ln{(\frac{K}{y})}}^{+\infty}[(ye^z-K)+(K-ye^z)]n(t,dz,y)+\int_{-\infty}^{\ln{(\frac{K}{y})}}[K-ye^z]n(t,dz,y)\nonumber\\
&=&\int_{-\infty}^{\ln{(\frac{K}{y})}}y(e^{\ln{(\frac{K}{y})}}-e^z)n(t,dz,y).\nonumber\\
\end{eqnarray*}
\end{itemize}
Using integration by parts, $\chi_{t,y}$  can be equivalently expressed as
\begin{equation}
  \chi_{t,y}(z)=
  \begin{cases}
&\int_{-\infty}^z (e^z-e^u)\,n(t,du,y), \quad z<0\ ;\nonumber\\
&\int_z^{\infty} (e^u-e^z)\,n(t,du,y), \:\quad z>0.\nonumber\\
\end{cases}
\end{equation}
Hence
$$
\int_{\mathbb{R}}[(ye^z-K)^+-e^z(y-K)^+-K(e^z-1)1_{\{ y>K\}}]n(t,dz,y)
=y\,\chi_{t,y}\left(\ln{\left(\frac{K}{y}\right)}\right).
$$
\end{proof}
%We will now present  two different derivations of Theorem 1:
%\begin{itemize}
% \item The first derivation (Section \ref{ch3.proof.mimick}) is based on a generalization \cite{bentatacont09} of Gy\"ongy's \cite{gyongy86} ``mimicking theorem" to discontinuous semimartingales: aside from the forward equation, this method also yields a { mimicking process} which has the same marginals as the price process.
% \item The second derivation (Section \ref{ch3.tanaka.sec}), based on the Tanaka-Meyer formula,  has the merit of requiring weaker assumptions on the local characteristics of the process.
%\end{itemize}

\subsection{Derivation of the forward equation}\label{ch3.tanaka.sec}
In this section we present a proof of Theorem
\ref{ch3.pide.forward.prop} using the Tanaka-Meyer formula
for semimartingales \cite[Theorem 9.43]{shengchia92}
under assumption \eqref{ch3.H}. %This proof does not require continuity assumptions on the coefficients of the forward PIDE.
\begin{proof}
We first note that, by replacing $\mathbb{P}$ by the conditional
measure ${\mathbb{P}}_{|\mathcal{F}_{t_0}}$  given
$\mathcal{F}_{t_0}$, we may replace the conditional expectation in
(\ref{ch3.def.call}) by an expectation with respect to the marginal
distribution $p^S_T(dy)$ of $S_T$ under
${\mathbb{P}}_{|\mathcal{F}_{t_0}}$. Thus, without loss of
generality, we set $t_0=0$ in the sequel and consider the case where
${\mathcal{F}_0}$ is the $\sigma$-algebra generated by all
$\mathbb{P}$-null sets and we denote $C_0(T,K)\equiv C(T,K)$ for
simplicity. (\ref{ch3.def.call}) can be expressed as
\begin{equation}\label{ch3.def.call.bis}
C(T,K)=e^{-\int_0^T r(t)\,dt}\int_{\mathbb{R}^+} \left(y-K\right)^+\,p^S_{T}(dy).
\end{equation}
By differentiating with respect to $K$, we obtain
\begin{equation}\label{ch3.breed}
  \begin{split}
    &\frac{\partial C}{\partial K}(T,K)=-e^{-\int_0^T r(t)\,dt}\int_K^\infty p^S_{T}(dy)=- e^{-\int_0^T r(t)\,dt}\mathbb{E}\left[1_{ \{S_{T}>K\}}\right],\\
    &\frac{\partial^2 C}{\partial K^2}(T,dy)=e^{-\int_0^T r(t)\,dt}p^S_{T}(dy).
  \end{split}
\end{equation}
Let $L^K_t=L^K_t(S)$ be the semimartingale local time of $S$ at $K$ under $\mathbb{P}$ (see \cite[Chapter 9]{shengchia92} or \cite[Ch. IV]{protter}
for definitions).
Applying the Tanaka-Meyer
formula  to $({S}_T-K)^+$, we have
\begin{equation}\label{ch3.tanaka}
\begin{split}
  (S_{T}-K)^+&=(S_0-K)^+ +\int_0^{T} 1_{ \{S_{t-}> K\}} dS_t + \frac{1}{2} (L^K_{T}) \\
  &+ \sum_{0< t\leq T} \left[ (S_{t}-K)^+-(S_{t-}-K)^+-1_{ \{S_{t-}> K\}}\Delta S_{t}\right].
\end{split}
\end{equation}
As noted in Remark \ref{ch3.novikov.rem}, the integrability condition \eqref{ch3.H} implies that the discounted  price
$\hat{S}_t=e^{-\int_0^t r(s)\,ds}S_t={\cal E}(U)_t$ is a
martingale under  $\mathbb{P}$.
So (\ref{ch3.stochmodel}) can be expressed as
$$dS_t=e^{\int_0^t r(s)\,ds}\left(r(t)S_{t-} dt+ d\hat{S}_t\right)$$
and
$$
  \int_0^{T} 1_{ \{S_{t-}> K\}} dS_t = \int_0^{T}e^{\int_0^t r(s)\,ds}\, 1_{ \{S_{t-}> K\}} d\hat{S}_t+
  \int_0^{T} e^{\int_0^t r(s)\,ds}\,r(t) S_{t-} 1_{ \{S_{t-}> K\}} dt,
$$
where the first term is a
martingale. Taking expectations, we obtain
 \begin{eqnarray}
e^{\int_0^{T}r(t)\,dt}C(T,K)-(S_0-K)^+&=&\mathbb{E}\left[\int_0^{T} e^{\int_0^t r(s)\,ds}\,r(t) S_t \,1_{\{ S_{t-}> K\}} dt +
        \frac{1}{2}L^K_{T}\right]  \nonumber\\[0.1cm]
    & + &  \mathbb{E}\left[ \sum_{0< t\leq T} \left((S_{t}-K)^+-(S_{t-}-K)^+-1_{\{ S_{t-}> K\}}\Delta S_{t}\right)\right].\nonumber
  \end{eqnarray}
Noting that $S_{t-}
1_{ \{S_{t-}> K\}}= (S_{t-}-K)^+ + K1_{ \{S_{t-}> K\}} $, we obtain
$$
\mathbb{E}\left[ \int_0^{T} e^{\int_0^t r(s)\,ds}\,r(t) S_{t-} 1_{ \{S_{t-}> K\}} dt\right]
 =\int_0^{T} r(t)e^{\int_0^{t}r(s)\,ds}\left[C(t,K)- K \frac{\partial C}{\partial K}(t,K)\right]\,dt,
$$
using Fubini's theorem and (\ref{ch3.breed}).
As for the jump term,
\begin{eqnarray*}
  &&\mathbb{E}\left[\sum_{0< t\leq T} (S_{t}-K)^+-(S_{t-}-K)^+-1_{ \{S_{t-}> K\}} \Delta S_{t}\right] \nonumber\\[0.1cm]
  &=& \mathbb{E}\left[\int_0^{T}dt\int m(t,dx)\,(S_{t-}e^x-K)^+-(S_{t-}-K)^+-1_{\{ S_{t-}> K\}}S_{t-}(e^x-1)\right]\nonumber\\[0.1cm]
  &=& \mathbb{E}\Big[\int_0^{T} dt \int m(t,dx)\big( (S_{t-}e^x-K)^+-(S_{t-}-K)^+ \nonumber\\[0.1cm]
    &&\quad\quad\quad-(S_{t-}-K)^+(e^x-1) - K1_{\{ S_{t-}> K\}}(e^x-1)\big)\Big]\nonumber.
\end{eqnarray*}
Applying Lemma \ref{ch3.calcul.de.B} to the random measure $m$ we
obtain that
$$
\int m(t,dx)\big((S_{t-}e^x-K)^+-e^{x}(S_{t-}-K)^+-K1_{\{ S_{t-}> K\}}(e^x-1)\big)=S_{t-}\,\psi_{t}\left(\ln{\left(\frac{K}{S_{t-}}\right)}\right)
$$
holds true. One observes that for all $z$ in $\mathbb{R}$
\begin{eqnarray*}
\psi_{t}(z)&\leq& 1_{\{z<0\}}\,\int_{-\infty}^z e^{z}\,m(t,du) +1_{\{z>0\}} \int_{-\infty}^z e^{u}\,m(t,du)\\
&=&1_{\{z<0\}}e^{z}\,\int_{-\infty}^z1.\, m(t,du) +1_{\{z>0\}} \int_{-\infty}^z e^{u}\,m(t,du).
\end{eqnarray*}
Using Assumption \eqref{ch3.H},
\begin{eqnarray*}
 \mathbb{E}\left[ \sum_{0< t\leq T} \left[ (S_{t}-K)^+-(S_{t-}-K)^+-1_{ \{S_{t-}> K\}}\Delta S_{t}\right]\right]
=\mathbb{E}\left[\int_0^{T} dt\,S_{t-}\,\psi_{t}\left(\ln{\left(\frac{K}{S_{t-}}\right)}\right)\right]<\infty.
\end{eqnarray*}
Hence applying Fubini's theorem leads to
\begin{eqnarray*}
 &&\mathbb{E}\left[\sum_{0< t\leq T} (S_{t}-K)^+-(S_{t-}-K)^+-1_{ \{S_{t-}> K\}} \Delta S_{t}\right]\\
  &=&\int_0^{T}dt\mathbb{E}\Big[\int \,m(t,dx)\big((S_{t-}e^x-K)^+
   -e^{x}(S_{t-}-K)^+-K1_{ \{S_{t-}> K\}}(e^x-1)\Big)\Big]\\
  &=& \int_0^{T}dt\,\mathbb{E}\left[S_{t-}\,\psi_{t}\left(\ln{\left(\frac{K}{S_{t-}}\right)}\right)\right]\nonumber\\[0.1cm]
   &=& \int_0^{T}dt\,\mathbb{E}\left[S_{t-}\mathbb{E}\left[\psi_{t}\left(\ln{\left(\frac{K}{S_{t-}}\right)}\right)|S_{t-}\right]\right]\nonumber\\[0.1cm]
&=&\int_0^{T}dt\,\mathbb{E}\left[S_{t-}\,\chi_{t,S_{t-}}\left(\ln{\left(\frac{K}{S_{t-}}\right)}\right)\right].
\end{eqnarray*}
Let $ \varphi\in C_0^\infty([0,T]\times]0,\infty[)$ be an infinitely differentiable function with compact support in $[0,T]\times]0,\infty[$. The extended occupation time formula  \cite[Chap. VI, Exercise 1.15]{revuzyor} yields
\begin{equation}\label{ch3.occup.time.formula}
\begin{split}
  \int_{0}^{+\infty} dK \int_0^{T}\,\varphi(t,K)\, dL^K_{t}&=\int_0^{T}\varphi(t,S_{t-}) d[S]_t^c= \int_0^{T} dt\,\varphi(t,S_{t-})S_{t-}^2\delta_t^2 .
\end{split}
\end{equation}
Since $ \varphi$ is bounded and has compact support, in order to apply Fubini's
theorem to
\begin{equation*}
\mathbb{E}\left[\int_{0}^{+\infty} dK \int_0^{T}\,\varphi(t,K)\,dL^K_{t}\right],
\end{equation*}
it is sufficient to show that $\mathbb{E}\left[L^K_{t}\right]<\infty$ for $t\in[0,T]$.
Rewriting equation (\ref{ch3.tanaka}) yields
\begin{equation*}
 \frac{1}{2} L^K_{T}= (S_{T}-K)^+-(S_0-K)^+-\int_0^{T} 1_{ \{S_{t-}> K\}} dS_t
  - \sum_{0< t\leq T} \left[ (S_{t}-K)^+-(S_{t-}-K)^+-1_{ \{S_{t-}> K\}}\Delta S_{t}\right].
\end{equation*}
 Since $\hat{S}$ is a martingale,   $\mathbb{E}[S_T]<\infty$ $\mathbb{E}\left[(S_T-K)^+\right]<\mathbb{E}\left[S_T\right]$ and \\
 $\mathbb{E}\left[\int_0^{T} 1_{ \{S_{t-}> K\}} dS_t\right]<\infty$. As discussed above,
$$\mathbb{E}\left[ \sum_{0< t\leq T}  \left((S_{t}-K)^+-(S_{t-}-K)^+-1_{ \{S_{t-}> K\}}\Delta S_{t}\right)\right]<\infty,$$
yielding that $\mathbb{E}\left[L^K_{T}\right]<\infty$. %Furthermore, since $ \varphi$ is bounded and has compact support, there exists $C>0$ such that,
%\begin{equation*}
%\begin{split}
%\int_T^{T+h} \mathbb{E}\left[ \varphi(S_{t-})S_{t-}^2\delta_t^2\right]\,dt&\leq C\,\int_T^{T+h} \mathbb{E}\left[ S_{t-}^2\delta_t^2 \right]\,dt\\
%&\leq C\,\sup_{[0,T]} \mathbb{E}\left[S_{t}^2\right]\left(\int_T^{T+h} \mathbb{E}\left[ \delta_t^2 \right]\,dt\right)<\infty.
%\end{split}
%\end{equation*}
%since Assumption \eqref{ch3.H}, allowing us to apply Fubini's theorem.
Hence, one may take expectations in equation (\ref{ch3.occup.time.formula}) to obtain
\begin{eqnarray}
\mathbb{E}\left[\int_{0}^{+\infty} dK \int_0^{T}\,\varphi(t,K)\,dL^K_{t}\right]&=& \mathbb{E}\left[ \int_0^{T}\varphi(t,S_{t-})S_{t-}^2\delta_t^2 dt\right]
  = \int_0^{T}dt\,\mathbb{E}\left[\varphi(t,S_{t-})S_{t-}^2\delta_t^2\right]\nonumber\\
  = \int_0^{T}dt\,\mathbb{E}\left[\mathbb{E}\left[\varphi(t,S_{t-})S_{t-}^2\delta_t^2 |S_{t-}\right] \right]
  &=& \mathbb{E}\left[\int_0^{T}dt\,\varphi(t,S_{t-})S_{t-}^2\sigma(t,S_{t-})^2 \right]\nonumber\\
  = \int_{0}^\infty \int_0^{T} \varphi(t,K)K^2\sigma(t,K)^2 p_t^S(dK)\,dt
  &=&  \int_0^{T} dt\,e^{\int_0^t r(s)\,ds}\int_{0}^\infty \varphi(t,K) K^2\sigma(t,K)^2 \frac{\partial^2 C}{\partial K^2}(t,dK),\nonumber
\end{eqnarray}
where the last line is obtained by using (\ref{ch3.breed}).
Using integration by parts,
 \begin{eqnarray*}
&& \int_{0}^\infty dK\, \int_0^{T} dt\,\varphi(t,K)\,\frac{\partial }{\partial t}\left[e^{\int_0^t r(s)\,ds}\,C(t,K)-(S_0-K)^+\right]\\[0.1cm]
&=& \int_{0}^\infty dK\, \int_0^{T} dt\,\varphi(t,K)\,\frac{\partial }{\partial t}\left[e^{\int_0^t r(s)\,ds}\,C(t,K)\right]\\[0.1cm]
&=& \int_{0}^\infty dK\, \int_0^{T} dt\,\varphi(t,K)\, e^{\int_0^{t} r(s)\,ds}\left[\frac{\partial C}{\partial t}(t,K)+r(t)C(t,K)\right]\\[0.1cm]
&=& -\int_{0}^\infty dK\, \int_0^{T} dt\,\frac{\partial \varphi}{\partial t} (t,K)\,\left[e^{\int_0^t r(s)\,ds}\,C(t,K)\right],
\end{eqnarray*}
where derivatives are used in the sense of distributions. Gathering together all terms,
\begin{eqnarray*}
&&\int_{0}^\infty dK\, \int_0^{T} dt\,\frac{\partial \varphi}{\partial t}(t,K)\,\left[e^{\int_0^t r(s)\,ds}\,C(t,K)\right]\\[0.1cm]
&=&\int_0^{T}dt\int_{0}^\infty dK\,\frac{\partial \varphi}{\partial t}(t,K)\,\int_0^t\,ds\,r(s)\,e^{\int_0^{s}r(u)\,du}[C(s,K)- K \frac{\partial C}{\partial K}(s,K)] \\
&+&\int_0^{T} dt\int_{0}^\infty \frac{1}{2}\frac{\partial \varphi}{\partial t}(t,K)\,\int_0^t dL_s^K\\[0.1cm]
&+&\int_0^{T}dt \int_{0}^{\infty}dK\,\,\frac{\partial \varphi}{\partial t}(t,K)\int_0^t ds\,e^{\int_0^{s} r(u)\,du}\,\int_0^{+\infty} y\frac{\partial^2 C}{\partial K^2}(s,dy)\chi_{s,y}\left(\ln{\left(\frac{K}{y}\right)}\right)\\[0.1cm]
&=&-\int_0^{T}dt\int_{0}^\infty dK\,\varphi(t,K)\,r(t)\,e^{\int_0^{t}r(s)\,ds}[C(t,K)- K \frac{\partial C}{\partial K}(t,K)] \nonumber\\[0.1cm]
&-&\int_0^{T} dt\int_{0}^\infty \frac{1}{2}\, \varphi(t,K)\,dL_t^K
-\int_0^{T}dt \int_{0}^{\infty}dK\,\varphi(t,K)\,e^{\int_0^{t} r(s)\,ds}\,\int_0^{+\infty} y\frac{\partial^2 C}{\partial K^2}(t,dy)\chi_{t,y}\left(\ln{\left(\frac{K}{y}\right)}\right).
\end{eqnarray*}
So finally we have shown that for any test function $ \varphi\in
C_0^\infty([0,T]\times]0,\infty[,\mathbb{R} )$,
\begin{eqnarray*}
&&\int_{0}^\infty dK\, \int_0^{T} dt\,\frac{\partial \varphi}{\partial t}(t,K)\,\left[e^{\int_0^t r(s)\,ds}\,C(t,K)\right]\\[0.1cm]
&=& -\int_{0}^\infty dK\, \int_0^{T} dt\,\varphi(t,K)\, e^{\int_0^{t} r(s)\,ds}\left[\frac{\partial C}{\partial t}(t,K)+r(t)C(t,K)\right]\\[0.1cm]
&=&-\int_0^{T}dt\int_{0}^\infty dK\,\varphi(t,K)\,r(t)\,e^{\int_0^{t}r(s)\,ds}[C(t,K)- K \frac{\partial C}{\partial K}(t,K)] \nonumber\\[0.1cm]
&-&\int_0^{T} dt\int_{0}^\infty \frac{1}{2}\,e^{\int_0^{t}r(s)\,ds}\,\varphi(t,K)\,K^2\sigma(t,K)^2 \frac{\partial^2 C}{\partial K^2}(t,dK) \nonumber\\[0.1cm]
&-&\int_0^{T}dt \int_{0}^{\infty}dK\,\varphi(t,K)\,e^{\int_0^{t} r(s)\,ds}\,\int_0^{+\infty} y\frac{\partial^2 C}{\partial K^2}(t,dy)\chi_{t,y}\left(\ln{\left(\frac{K}{y}\right)}\right).
\end{eqnarray*}
Therefore, $C(.,.)$ is a solution of \eqref{ch3.pide.forward.eq}
in the sense of distributions.
\end{proof}

\subsection{Uniqueness of solutions of the forward PIDE}
Theorem \ref{ch3.pide.forward.prop} shows that the call price
$(T,K)\mapsto C_{t_0}(T,K)$ solves the forward PIDE
\eqref{ch3.pide.forward.eq}.   Uniqueness of the solution of such PIDEs
 has been shown using analytical methods \cite{barles08,MR1911531} under various types of
conditions on the coefficients. We give below a direct proof of
uniqueness for \eqref{ch3.pide.forward.eq} using a probabilistic method,
under explicit conditions which cover most examples of models used
in finance.

Define, for   $u\in\mathbb{R}, t\in [0,T[, z> 0$ the measure
$n(t,du,z)$ by
\begin{equation}
\begin{split}
n(t,[u,\infty[,z)&=-e^{-u}\,\frac{\partial}{\partial u}\left[\chi_{t,z}(u)\right],\quad u>0\ ;\\
n(t,]-\infty,u],z)&=e^{-u}\frac{\partial}{\partial u}\left[\chi_{t,z}(u)\right],\quad u<0.
\end{split}
\end{equation}
Throughout this section, we make the following assumption:

\begin{assumption}\label{ch3.Assumption3}
\begin{equation}
\forall T>0, \forall B\in\mathcal{B}(\mathbb{R})-\{0\},\quad
(t,z)\to\sigma(t,z),\quad (t,z)\to n(t,B,z)\nonumber\end{equation}
are continuous in $z\in\mathbb{R}^+$, uniformly in $t\in[0,T]$; right-continuous in $t$ on $[0,T[$ uniformly
in $z\in \mathbb{R}^+$. and
\begin{equation}\tag{H'}\label{ch3.H'} \exists K_T>0,
\forall (t,z)\in[0,T]\times\mathbb{R}^+,\
|\sigma(t,z)|+\int_{\mathbb{R}} (1\wedge |z|^2)\,n(t,du,z)\leq
K_T.\end{equation}
\end{assumption}
%Note that \eqref{ch3.H'} implies our previous assumption \eqref{ch3.H}.
\begin{theorem}\label{ch3.forward.pide.unicity}
Under Assumption \ref{ch3.Assumption3}, if
\begin{eqnarray*}
  %&&\quad (i)\quad (t,z)\to\sigma(t,z)\, \mathrm{is\,continuous\,in}\, z\in\mathbb{R}^+\,\mathrm{uniformly \,in}\,t\in[0,T]\\
  &\mathrm{either}&\quad (i)\quad \forall R>0\,\ \forall t\in [0,T[, \quad \inf_{0\leq z \leq R}\sigma(t,z)>0, \\
  &\mathrm{or}&\quad(ii)\quad \sigma(t,z)\equiv 0\quad{\rm and}\ \exists \beta\in]0,2[,\, \exists C>0, \forall R>0, \forall (t,z)\in [0,T[\times[0,R],\\[0.1cm]
  && \quad \,\forall f\in{C}^0_0(\mathbb{R}-\{0\},\mathbb{R}_+),\qquad\int \left(n(t,du,z)-\frac{C\,du}{|u|^{1+\beta}}\right)\,f(u)\geq 0,\\[0.1cm]
  && \quad\,\exists K'_{T,R}>0, \int_{\{|u|\leq 1\}}|u|^{\beta}\,\left(n(t,du,z)-\frac{C\,du}{|u|^{1+\beta}}\right)\,dt\leq K'_{T,R}, \\
  &\mathrm{and}&\quad(iii)\quad \lim_{R\to\infty} \int_0^T \sup_{z\in\mathbb{R}^+} n\left(t,\{|u|\geq R\},z\right)\,dt=0,
\end{eqnarray*}
then the call
option price $(T,K)\mapsto C_{t_0}(T,K)$, as a function of maturity and
 strike, is the unique solution (in the sense of distributions) of the partial
integro-differential equation (\ref{ch3.pide.forward.eq}) on $[t_0,\infty[\times ]0,\infty[$ with
the initial condition: $\forall K>0\quad C_{t_0}(t_0,K)= (S_{t_0}-K)_+.$
\end{theorem}

%Define for all $B\in\mathcal{B}(\mathbb{R})-\{0\}$, for all $t\in
%[0,T]$ and $z>0$
%\begin{equation}
%n(t,B,z)=\mathbb{E}^{\mathbb{P}}\left[m(.,t,B)|S_{t^-}=z\right]
%\end{equation}

The proof uses the uniqueness of the solution of the forward
Kolmogorov equation associated to a certain
integro-differential operator. We start by recalling the following result:
%developped in Chapter \ref{chapter.mimic}.
\begin{proposition}\label{ch3.evolution.equation}
Define, for $t\in[0,T]$ and $f\in C^\infty_0(\mathbb{R})$, the
integro-differential operator
\begin{equation}
  \begin{split}
    L_tf(x)&=r(t)xf'(x)+\frac{x^2\sigma(t,x)^2}{2}f''(x)\\
    &+\int_{\mathbb{R}}\left[f(t,xe^y)-f(t,x)-x(e^y-1) f'(x)\right]n(t,dy,x).
  \end{split}
\end{equation}
Under Assumption \ref{ch3.Assumption3}, if {\it either} conditions $(i)$
{\it or} $(ii)$  and $(iii)$ of Theorem \ref{ch3.forward.pide.unicity}
hold, then for each $x_0$ in $\mathbb{R}^+$, there exists a unique
family $(p_t(x_0,dy),t\geq 0)$ of bounded measures such that
\begin{equation}\label{ch3.kolmogore}
\forall g\in  \mathcal{C}_0^\infty(]0,\infty[,\mathbb{R}),\quad\int
g(y)\frac{dp}{dt}(x_0,dy)=\int p_t(x_0,dy)L_tg(y),\quad
p_0(x_0,.)=\epsilon_{x_0},
\end{equation}
where $\epsilon_{x_0}$ is the point mass at $x_0$. Furthermore,
$p_t(x_0,.)$ is a probability measure on $[0,\infty[$.
\end{proposition}
\begin{proof}
%Given Assumption \ref{ch3.Assumption3} and assumptions $(i)$ (or $(ii)$) and
%$(iii)$, one may apply directly Theorem \ref{evolution.equation}.
Denote by $(X_t)_{t\in [0,T]}$ the canonical process on
$D([0,T],\mathbb{R}_+)$. Under assumptions $(i)$ (or $(ii)$) and
$(iii)$, the martingale problem for $((L_t)_{t\in[0,T]},\mathcal{C}_0^\infty(\mathbb{R}^+))$  on $[0,T]$ is well-posed \cite[Theorem 1]{mikulevicius90}: for any
$x_0\in\mathbb{R}^+, t_0\in [0,T[$, there exists a unique
probability measure $\mathbb{Q}_{t_0,x_0}$ on
$(D([0,T],\mathbb{R}^+),{\cal B}_T)$ such that
$\mathbb{Q}_{t_0,x_0}(X_{t_0}=x_0)=1$ and for any
$f\in\mathcal{C}_0^\infty(\mathbb{R}^+)$,
\[ f(X_t)-f(x_0)-\int_{t_0}^t L_sf(X_s)\,ds\]
is a $\left(\mathbb{Q}_{t_0,x_0},(\mathcal{B}_t)_{t\geq t_0}\right)$-martingale on $[t_0,T]$.
Under $\mathbb{Q}_{t_0,x_0}$, $X$ is a  Markov process. Define the
evolution operator $(Q_{t_0,t})_{t\in[t_0,T]}$ by
\begin{equation}
\forall f\in \mathcal{C}_b^0(\mathbb{R}^+),\quad Q_{t_0,t}
f(x_0)=\mathbb{E}^{\mathbb{Q}_{t_0,x_0}}\left[f(X_t)\right].
\end{equation}
\begin{eqnarray*}
{\rm Then}\qquad Q_{t_0,t} f(x_0)
&=&f(x_0)+\mathbb{E}^{\mathbb{Q}_{t_0,x_0}}\left[\int_{t_0}^t L_sf(X_s)\,ds\right].\\
\end{eqnarray*}
Given Assumption \ref{ch3.Assumption3}, $t\in[0,T]\mapsto \int_{t_0}^t L_sf(X_s)\,ds$ is uniformly bounded on $[0,T]$. Given Assumption \ref{ch3.Assumption3}, since $X$ is right continuous $s\in[0,T[\mapsto L_sf(X_s)$ is right-continuous up to a $\mathbb{Q}_{t_0,x_0}$-null set and
\begin{equation*}
\lim_{t\downarrow t_0} \int_{t_0}^t L_sf(X_s)\,ds= 0 \quad\mathrm{a.s.}
\end{equation*}
Applying the dominated convergence theorem yields
\begin{equation*}
\lim_{t\downarrow t_0} \mathbb{E}^{\mathbb{Q}_{t_0,x_0}}\left[\int_{t_0}^t L_sf(X_s)\,ds\right]=0,
\qquad {\rm so}\quad
\lim_{t\downarrow t_0} Q_{t_0,t} f(x_0) = f(x_0),
\end{equation*}
implying that $t\in[0,T[\mapsto Q_{t_0,t} f(x_0)$ is right-continuous at $t_0$ for each $x_0\in\mathbb{R}^+$. Hence the evolution operator $(Q_{t_0,t})_{t\in[t_0,T]}$ verifies the following continuity property:
\begin{equation}
\forall f\in \mathcal{C}_b^0(\mathbb{R}^+), \forall x\in
 \mathbb{R}^+,\quad\lim_{t\downarrow t_0} Q_{t_0,t}f (x) =
 f(x).\label{ch3.strong.continuity}
\end{equation}
In particular, denoting $q_t(dy)$ the marginal distribution  of $X_t$, the map
\begin{equation}\label{ch3.weak-right-continuity}
 t\in[0,T[\mapsto \int_{\mathbb{R}^+}
q_t(dy) f(y)
\end{equation}
is right-continuous, for any $f\in \mathcal{C}_b^0(\mathbb{R}^+), x_0\in\mathbb{R}^+$.
The martingale property implies that $q_t(x_0,dy)$ satisfies
\begin{equation}\label{ch3.eq.qt.bis}
\forall g\in\mathcal{C}_0^\infty(\mathbb{R}^+),\quad \int_{\mathbb{R}^+}  q_t(x_0,dy)g(y)=g(x_0)+\int_0^t\int_{\mathbb{R}^+} q_s(x_0,dy)L_sg(y)\,ds.
\end{equation}
%The right-continuity of $t\in[0,T[\mapsto \int_{\mathbb{R}^+}  q_t(x_0,dy)g(y)$ for any $g$ in $\mathcal{C}_b^0(\mathbb{R}^+)$.
Given Assumption \ref{ch3.Assumption3}, $q_t$ is a solution of \eqref{ch3.kolmogore} with initial condition $q_0(dy)=\epsilon_{x_0}$.
In particular, the measure $q_t$ has  mass $1$.
To show uniqueness of solutions of \eqref{ch3.kolmogore},  we will rewrite \eqref{ch3.kolmogore} as the forward Kolmogorov equation
associated with a {\it homogeneous}  operator on space-time domain and use uniqueness results for the corresponding homogeneous equation.
Let $\mathcal{C}^1([0,T])\otimes \mathcal{C}_0^\infty(\mathbb{R}^+)$ be the tensor product of $\mathcal{C}^1([0,T])$ and $\mathcal{C}_0^\infty(\mathbb{R}^+)$. $L_t$ can be extended to a (homogeneous) linear operator $A$ on $\mathcal{C}^1([0,T])\otimes \mathcal{C}_0^\infty(\mathbb{R}^+)$ defined via
%Let  $f\in C^\infty_0(\mathbb{R}^d)$ and $\gamma\in\mathcal{C}^1([0,T])$ and
%consider the (homogeneous) dependent operator $A$ mapping functions of the form $(t,x)\in[0,T]\times\mathbb{R}^d\to f(x)\gamma(t)$, into :
% which will be denoted $\mathcal{C}^1([0,T])\otimes \mathcal{C}_0^\infty(\mathbb{R}^d)$
\begin{equation}\label{ch3.new.operator}
\forall f\in C^\infty_0(\mathbb{R}^+),\quad\forall \gamma\in\mathcal{C}^1([0,T]),\quad A(f\gamma)(t,x)=\gamma(t)L_tf(x)+f(x)\gamma'(t).
\end{equation}
\cite[Theorem 7.1, Chapter 4]{ethierkurtz} implies that for any $x_0\in \mathbb{R}^+$, if  $(X,\mathbb{Q}_{t_0,x_0})$ is a solution of the
martingale problem for $L$, then the law of $\eta_t=(t,X_t)$ under $\mathbb{Q}_{t_0,x_0}$ is a
solution of the martingale problem for  $A$: in particular for any
$f\in\mathcal{C}_0^\infty(\mathbb{R}^+)$ and
$\gamma\in\mathcal{C}([0,T])$,
\begin{equation}
     \int q_t(x_0,dy)f(y)\gamma(t)= f(x_0)\gamma(0)+\int_0^t \int q_s(x_0,dy)A(f\gamma)(s,y)\,ds.
    \end{equation}
 \cite[Theorem 7.1, Chapter 4]{ethierkurtz} implies also that if the law of $\eta_t=(t,X_t)$ is a
solution of the martingale problem for $A$ then the law of $X$ is also a solution of the
martingale problem for $L$, namely: uniqueness holds for the martingale problem associated to the
operator $L$ on $\mathcal{C}_0^\infty(\mathbb{R}^+)$ if and only if uniqueness
holds for the martingale problem associated to the  martingale problem for $A$ on
$\mathcal{C}^1([0,T])\otimes C^\infty_0(\mathbb{R}^+)$.
Define, for  $t\in[0,T] $ and $h \in \mathcal{C}_b^0([0,T]\times\mathbb{R}^+)$,
\begin{equation}
\forall (s,x) \in[0,T] \times \mathbb{R}^+,\quad \mathcal{Q}_{t}h(s,x)=\int_{\mathbb{R}^+} \,q_{t}(x,dy)h(t,y).
\end{equation}
which extends $Q_{0,t}$  to a 'homogeneous' operator on   $\mathcal{C}_b^0([0,T]\times\mathbb{R}^+)$. Using  (\ref{ch3.eq.qt.bis}), we have, for  $\epsilon>0$,
 \begin{eqnarray}\label{ch3.semigroup}
\forall (f,\gamma)\in\mathcal{C}^1([0,T])\times \mathcal{C}_0^\infty(\mathbb{R}^+), \quad
 \mathcal{Q}_{t}(f\gamma)(s,x_0)-\mathcal{Q}_{\epsilon}(f\gamma)(s,x_0) =\nonumber\\
\int_{\epsilon}^{t} \int_{\mathbb{R}^+} q_u(x_0,dy)A(f\gamma)(u,y)\,du
=\int_{\epsilon}^{t} \mathcal{Q}_{u}(A(f\gamma))(s,x_0)\,du .
\end{eqnarray}
By linearity,  for any $h\in\mathcal{C}^1([0,T])\otimes \mathcal{C}_0^\infty(\mathbb{R}^+)$ we have
\begin{equation}
 \mathcal{Q}_{t}h(s,x_0)-\mathcal{Q}_{\epsilon}h(s,x_0) =\int_{\epsilon}^{t} \int_{\mathbb{R}^+} q_u(x_0,dy)Ah(u,y)\,du=\int_{\epsilon}^{t} \mathcal{Q}_{u}Ah(s,x_0)\,du.
\end{equation}
%so using the fact that $\mathcal{C}^1([0,T])\otimes \mathcal{C}_0^\infty(\mathbb{R}^d) $ is dense in $\mathcal{C}_b^0([0,T]\times\mathbb{R}^d)$ for the uniform norm, we obtain that for any $ h\in\mathcal{C}_b^0([0,T]\times\mathbb{R}^d)$,
%$(\mathcal{Q}_t)_{t\geq 0}$ thus defines a strongly continuous semigroup on $\mathcal{C}_b^0([0,T]\times \mathbb{R}^d)$.
Consider now a family  $p_t(x_0,dy)$ of positive measures  solution of (\ref{ch3.kolmogore}) such that $p_0(x_0,dy)=\epsilon_{x_0}(dy)$. Then $p_t$ is also a solution of (\ref{ch3.eq.qt.bis}). An integration  by parts implies that, for $ (f,\gamma)\in\mathcal{C}^1([0,T])\times \mathcal{C}_0^\infty(\mathbb{R}^+),$
    \begin{equation}\label{ch3.eq.pt.bis}
      \int_{\mathbb{R}^+} p_t(x_0,dy)f(y)\gamma(t)= f(x_0)\gamma(0)+\int_0^t \int_{\mathbb{R}^+} p_s(x_0,dy)A(f\gamma)(s,y)\,ds.
    \end{equation}
Define, for $t$ in $[0,T], h\in\mathcal{C}_b^0([0,T]\times\mathbb{R}^+)$,
\begin{eqnarray*}
\forall (s,x_0)\in [0,T]\times \mathbb{R}^+,\qquad \mathcal{P}_{t} h(s,x_0)&=& \int_{\mathbb{R}^+} p_t(x_0,dy)h(t,y).
\end{eqnarray*}
$(\mathcal{P}_{t})_{t\geq 0}$ is then a homogeneous semigroup.
Using (\ref{ch3.eq.pt.bis}), for $(f,\gamma)\in\mathcal{C}^1([0,T])\times \mathcal{C}_0^\infty(\mathbb{R}^+)$,
\begin{equation}\label{ch3.semigroup.bis}
\forall \epsilon>0,\quad \mathcal{P}_{t}(f\gamma)-\mathcal{P}_{\epsilon}(f\gamma) =\int_{\epsilon}^{t} \int_{\mathbb{R}^+} p_u(dy)A(f\gamma)(u,y)\,du=\int_{\epsilon}^{t} \mathcal{P}_{u}(A(f\gamma))\,du.
\end{equation}
which is identical to \eqref{ch3.semigroup}. Multiplying by $e^{-\lambda t}$ and integrating with respect to $t$ we obtain
\begin{eqnarray*}
\lambda \int_0^\infty e^{-\lambda t} \,\mathcal{P}_{t}(f\gamma)(s,x_0)\,dt&=&f(x_0)\gamma(0)+\lambda\int_0^\infty e^{-\lambda t}\int_0^t \mathcal{P}_{u}(A(f\gamma))(s,x_0)\,du\,dt\\
&=&f(x_0)\gamma(0)+\lambda\int_0^\infty \left(\int_u^\infty e^{-\lambda t} dt\right)\, \mathcal{P}_{u}(A(f\gamma))(s,x_0)\,du\\
&=&f(x_0)\gamma(0)+\int_0^\infty e^{-\lambda u}\, \mathcal{P}_{u}(A(f\gamma))(s,x_0)\,du.
\end{eqnarray*}
for any $\lambda>0$. Similarly, from  \eqref{ch3.semigroup} we obtain for any $\lambda>0,$
\begin{eqnarray*}
\lambda \int_0^\infty e^{-\lambda t} \,\mathcal{Q}_{t}(f\gamma)(s,x_0)\,dt
&=&f(x_0)\gamma(0)+\int_0^\infty e^{-\lambda u}\, \mathcal{Q}_{u}(A(f\gamma))(s,x_0)\,du.
\end{eqnarray*}
Hence for $(f,\gamma)\in\mathcal{C}^1([0,T])\times \mathcal{C}_0^\infty(\mathbb{R}^+)$ we have
\begin{equation}\label{ch3.identity.semigroup}
\int_0^\infty e^{-\lambda t}\,\mathcal{Q}_{t}(\lambda-A)(f\gamma)(s,x_0)\,dt=f(x_0)\gamma(0)=\int_0^\infty e^{-\lambda t}\,\mathcal{P}_{t}(\lambda- A)(f\gamma)\,dt.
\end{equation}
By linearity,  for any $h\in\mathcal{C}^1([0,T])\otimes \mathcal{C}_0^\infty(\mathbb{R}^+)$ we have
\begin{equation}\label{ch3.identity.semigroup.bis}
\int_0^\infty e^{-\lambda t}\,\mathcal{Q}_{t}(\lambda-A)h(s,x_0)\,dt=h(0,x_0)=\int_0^\infty e^{-\lambda t}\,\mathcal{P}_{t}(\lambda- A)h\,dt.
\end{equation}
%since $\int q_{0,x_0}(dy)\mu_0(dx_0)=p_{0}(dy)$ where $\mu_0=\mathcal{L}(X_0)$. \\
%Since the functions of the form $f\gamma$ generate the space $C^\infty_0([0,T]\times\mathbb{R}^d)$ and $t\to\beta^Y$, $t\to a^Y$, and $t\to m^Y(t,A)$ for $A\in\mathcal{B}(\mathbb{R}^d) -\{0\}$ are bounded in $t$ on $[0,T]$, it implies that for any fixed  $f\in \mathcal{C}^\infty_0(\mathbb{R}^d)$ and any fixed $\gamma\in\mathcal{C}^1([0,T])$, $t\to Q_tL^0(f\gamma)$ is bounded on $[0,T]$, which generalizes (\ref{identity.semigroup}) into
%\begin{equation}\label{identity.semigroup.bis}
%\int_0^\infty e^{-\lambda t}\,Q_{t}(\lambda- L^0)g\,(X_0)\,dt=Q_{0}g\,(X_0)=\int_0^\infty e^{-\lambda t}\,P_{t}(\lambda- L^0)g \,dt
%\end{equation}
%Since $Q_t$ is  strongly continuous on $\mathcal{C}_b^0(\mathbb{R}^d)$,
%and Theorem 2.6
From \cite[Proposition 2.1, Chapter 1]{ethierkurtz}, for all $\lambda>0$ $$Im(\lambda-A)=\mathcal{C}_b^0([0,T]\times\mathbb{R}^+)$$ where $Im(\lambda-A)$ denotes the image of $\mathcal{C}^1([0,T])\otimes \mathcal{C}_0^\infty(\mathbb{R}^+)$ by the mapping $(\lambda- A)$. Hence, since (\ref{ch3.identity.semigroup.bis}) holds
\begin{equation}\label{ch3.identity.semigroup.tris}
\forall h_in\mathcal{C}_b^0([0,T]\times \mathbb{R}^+),\quad\int_0^\infty e^{-\lambda t}\,\mathcal{Q}_{t}h\,(s,x_0)\,dt=\int_0^\infty e^{-\lambda t}\,\mathcal{P}_{t}h(s,x_0) \,dt,
\end{equation}
so the Laplace transform of $t\mapsto \mathcal{Q}_{t}h\,(s,x_0)$ is uniquely determined.
Using \eqref{ch3.semigroup.bis},
\begin{eqnarray}\label{ch3.semigroup.tris}
\forall \epsilon>0, \forall h\in\mathcal{C}^1([0,T])\otimes \mathcal{C}_0^\infty(\mathbb{R}^+),\nonumber\\
 \mathcal{P}_{t}h-\mathcal{P}_{\epsilon}h =\int_{\epsilon}^{t} \int_{\mathbb{R}^+} p_u(dy)Ah(u,y)\,du=\int_{\epsilon}^{t} \mathcal{P}_{u}(Ah)\,du
\end{eqnarray}
by linearity, which allows to show that, for any $h\in\mathcal{C}^1([0,T])\otimes \mathcal{C}_0^\infty(\mathbb{R}^+),$
$ t\mapsto \mathcal{P}_{t}h(s,x_0)$ is right-continuous:
$$
\forall h\in\mathcal{C}^1([0,T])\otimes \mathcal{C}_0^\infty(\mathbb{R}^+),\qquad \lim_{t' \downarrow t} \mathcal{P}_{t'}h(s,x_0)=\mathcal{P}_th(s,x_0).
$$
An identical argument using \eqref{ch3.semigroup.bis} shows that $
t\mapsto \mathcal{Q}_{t}h(s,x_0)$ is right-continuous. These two
right-continuous functions have the same Laplace transform by
\eqref{ch3.identity.semigroup.tris}, so they are equal. Thus we have
shown that
\begin{equation}
\forall  h\in\mathcal{C}^1([0,T])\otimes
\mathcal{C}_0^\infty(\mathbb{R}^+),\quad \int h(t,y) q_{t}(x_0,dy)=
\int h(t,y) p_t(x_0,dy). \label{ch3.eq.uniqueness}
\end{equation}
\cite[Proposition 4.4, Chapter 3]{ethierkurtz} implies that
$\mathcal{C}^1([0,T])\otimes \mathcal{C}_0^\infty(\mathbb{R}^+)$ is
separating, so \eqref{ch3.eq.uniqueness} allows to conclude that
$p_t(x_0,dy)=q_t(x_0,dy)$.
\end{proof}

We can now study the uniqueness of the forward PIDE
(\ref{ch3.pide.forward.eq}) and prove Theorem \ref{ch3.forward.pide.unicity}.
\begin{proof} of Theorem \ref{ch3.forward.pide.unicity}.
We start by decomposing $L_t$ as $L_t=A_t+B_t$ where
\begin{eqnarray*}
A_tf(y)=r(t)yf'(y)+\frac{y^2\sigma(t,y)^2}{2}f''(y),\quad {\rm and} \\
B_tf(y)=\int_{\mathbb{R}}[f(ye^z)-f(y)-y(e^z-1)f'(y)]n(t,dz,y).
\end{eqnarray*}
Then using the fact that $y\frac{\partial}{\partial y}(y-x)^+= x 1_{\{y>x\}}+(y-x)_+=y\,1_{\{y>x\}}$ and $\frac{\partial^2}{\partial y^2}(y-x)^+=\epsilon_x(y)$ where $\epsilon_x$ is a unit mass at $x$, we obtain
\begin{equation*}
A_t(y-x)^+=r(t)y\,1_{\{y>x\}}+\frac{y^2\sigma(t,y)^2}{2}\epsilon_x(y)\qquad {\rm and}
\end{equation*}
\begin{equation*}
\begin{split}
B_T(y-x)^+&=\int_{\mathbb{R}}[(ye^z-x)^+-(y-x)^+-(e^z-1)\left(x1_{\{y>x\}}+(y-x)^+\right)]n(t,dz,y)\\
&=\int_{\mathbb{R}}[(ye^z-x)^+-e^z(y-x)^+-x(e^z-1)1_{\{y>x\}}]n(t,dz,y).
\end{split}
\end{equation*}
Using Lemma \ref{ch3.calcul.de.B} for the random measure $n(t,dz,y)$ and $\psi_{t,y}$ its exponential double tail,
\begin{equation*}
B_t(y-x)^+=y\psi_{t,y}\left(\ln{\left(\frac{x}{y}\right)}\right)
\end{equation*}
Hence, the following identity holds
\begin{equation}\label{ch3.gen.payoff}
L_t(y-x)^+=r(t)\left(x 1_{\{y>x\}}+(y-x)_+\right)+\frac{y^2\sigma(t,y)^2}{2}\epsilon_x(y)+
y\psi_{t,y}\left(\ln{\left(\frac{x}{y}\right)}\right).
\end{equation}
Let $f: [t_0,\infty[\times ]0,\infty[\mapsto \mathbb{R}$ be a
solution in the sense of distributions of \eqref{ch3.pide.forward.eq}
with the initial condition : $f(0,x)=(S_0-x)^+$. Integration by
parts yields
\begin{eqnarray*}
&&\int_0^\infty \frac{\partial^2 f}{\partial x^2}(t,dy)L_t(y-x)^+\\
&=&\int_0^\infty \frac{\partial^2 f}{\partial x^2}(t,dy)\left(r(t)(x 1_{\{y>x\}}+(y-x)_+)+\frac{y^2\sigma(t,y)^2}{2}\epsilon_x(y)+
y\psi_{t,y}\left(\ln{\left(\frac{x}{y}\right)}\right)\right)\\
&=&-r(t)x\int_0^\infty \frac{\partial^2 f}{\partial x^2}(t,dy) 1_{\{y>x\}} + r(t)\int _0^\infty \frac{\partial^2 f}{\partial x^2}(t,dy)(y-x)^+\\
&+&\frac{x^2\sigma(t,x)^2}{2}\frac{\partial^2 f}{\partial x^2}+\int_0^\infty \frac{\partial^2 f}{\partial x^2}(t,dy) y\psi_{t,y}\left(\ln{\left(\frac{x}{y}\right)}\right)\\
&=&-r(t)x\frac{\partial f}{\partial x}+r(t)f(t,x)+\frac{x^2\sigma(t,x)^2}{2}\frac{\partial^2 f}{\partial x^2}+\int_0^\infty \frac{\partial^2 f}{\partial x^2}(t,dy)\, y\,\psi_{t,y}\left(\ln{\left(\frac{x}{y}\right)}\right).
\end{eqnarray*}
Hence given (\ref{ch3.pide.forward.eq}), the following equality holds
\begin{equation}\label{ch3.pide.forward.eq.gen}
\frac{\partial f}{\partial t}(t,x)=-r(t)f(t,x)+\int_0^\infty \frac{\partial^2 f}{\partial x^2}(t,dy)L_t(y-x)^+\ ,
\end{equation}
or, equivalently, after integration with respect to time $t$
\begin{equation}\label{ch3.f.integre}
 e^{\int_0^t r(s)\,ds}f(t,x)-f(0,x)= \int_0^\infty e^{\int_0^t r(s)\,ds}\frac{\partial^2 f}{\partial x^2}(t,dy)L_t(y-x)^+.
\end{equation}
 Integration by parts shows that
\begin{equation}
f(t,x)=\int_0^\infty  \frac{\partial^2 f}{\partial x^2}(t,dy)(y-x)^+.
\end{equation}
Hence (\ref{ch3.pide.forward.eq.gen}) may be rewritten as
\begin{equation}\label{ch3.pide.forward.eq.gen.integrated}
\int_0^\infty e^{\int_0^t r(s)\,ds}\frac{\partial^2 f}{\partial x^2}(t,dy)(y-x)^+-(S_0-x)^+= \int_0^t\int_0^\infty e^{\int_0^s r(u)\,du}\frac{\partial^2 f}{\partial x^2}(s,dy)L_s(y-x)^+\,ds.
\end{equation}
%We intend to prove that if $f$ satisfies (\ref{ch3.pide.forward.eq.gen}) with the initial condition $f(0,x)=(S_0-x)^+$, then
Define $q_t(dy)\equiv e^{\int_0^t r(s)\,ds}\,\frac{\partial^2 f}{\partial x^2}(t,dy)$, we have $q_0(dy)=\epsilon_{S_0}(dy)=p_{0}(S_0,dy)$.
For $g\in\mathcal{C}_0^\infty(]0,\infty[,\mathbb{R})$,  integration by parts yields
\begin{equation}
g(y)=\int_0^\infty g''(z)(y-z)^+\,dz.
\end{equation}
Replacing the above expression in $\int_{\mathbb{R}} g(y)q_t(dy)$ and using (\ref{ch3.pide.forward.eq.gen.integrated}) we obtain
\begin{eqnarray*}
&&\int_0^\infty g(y)q_t(dy)=\int_0^\infty g(y)\, e^{\int_0^t r(s)\,ds}\,\frac{\partial^2 f}{\partial x^2}(t,dy)\\
&=&\int_0^\infty g''(z)\int_0^\infty e^{\int_0^t r(s)\,ds}\,\frac{\partial^2 f}{\partial x^2}(t,dy)(y-z)^+\,dz\\
&=&\int_0^\infty g''(z)(S_0-z)^+\,dz+ \int_0^\infty g''(z)\int_0^t\int_0^\infty e^{\int_0^s r(u)\,du}\frac{\partial^2 f}{\partial x^2}(s,dy)L_s(y-z)^+\,dz\\
&=&g(S_0)+\int_0^t\int_0^\infty e^{\int_0^s r(u)\,du}\frac{\partial^2 f}{\partial x^2}(s,dy)L_s [\int_0^\infty g''(z)(y-z)^+\,dz]\\
&=&g(S_0)+\int_0^t\int_0^\infty q_s(dy)L_sg(y)\,ds.
\end{eqnarray*}
This is none other than equation (\ref{ch3.kolmogore}). By uniqueness of the solution $p_t(S_0,dy)$ of \eqref{ch3.kolmogore} in Proposition \ref{ch3.evolution.equation}),
$$
e^{\int_0^t r(s)\,ds}\,\frac{\partial^2 f}{\partial
x^2}(t,dy)= p_t(S_0,dy).
$$
One may rewrite equation \eqref{ch3.f.integre} as
\begin{equation*}
f(t,x)=  e^{-\int_0^t r(s)\,ds} \left(f(0,x)+ \int_0^\infty p_t(S_0,dy) L_t(y-x)^+\right),
\end{equation*}
showing that the solution
of   (\ref{ch3.pide.forward.eq}) with initial condition $f(0,x)=(S_0-x)^+$ is unique.
\end{proof}

\section{Examples}\label{ch3.examples.sec}
We now give various examples of  pricing models for which Theorem \ref{ch3.pide.forward.prop} allows to  retrieve or generalize previously known forms of forward pricing equations.
\subsection{It\^o processes}
When $(S_t)$ is an It\^o process i.e. when the jump part is absent,
the forward equation (\ref{ch3.pide.forward.eq}) reduces to the Dupire
equation \cite{dupire94}. In this case our result reduces to the
following:
\begin{proposition}[Dupire equation]
Consider the price process $(S_t)$ whose dynamics under the pricing measure
$\mathbb{P}$ is given by
$$
S_T = S_0 +\int_0^T r(t) S_{t} dt + \int_0^T S_{t}\delta_t dW_t.
$$
Assume there exists a measurable function  $\sigma:[t_0,T]\times \mathbb{R}^+-\{0\}\mapsto \mathbb{R}^+$ such that
\begin{equation}
 \forall t\in t\in[t_0,T],\qquad \sigma(t,S_{t-})=\sqrt{\mathbb{E}\left[\delta_t^2|S_{t^-}\right]}.
\end{equation}
If
\begin{equation}
\mathbb{E}\left[\exp{\left(\frac{1}{2}\int_0^T \delta_t^2\,dt\right)}\right]<\infty\qquad\ a.s.\quad \label{ch3.A1a}
\end{equation}
 the call
option price \eqref{ch3.def.call}
 is a solution (in the sense of distributions) of the partial differential equation
\begin{equation}
  \frac{\partial C_{t_0}}{\partial
T}(T,K) = -r(T)K\frac{\partial C_{t_0}}{\partial K}(T,K)+\frac{K^2\sigma(T,K)^2}{2}\, \frac{\partial^2 C_{t_0}}{\partial
K^2}(T,K)
\end{equation}
on $[t_0,\infty[\times ]0,\infty[$ with
the initial condition:
$$\forall K>0,\quad C_{t_0}(t_0,K)= (S_{t_0}-K)_+.$$
\end{proposition}
Notice in particular that this result does not require a
non-degeneracy condition on the diffusion term.
\begin{proof}
It is sufficient to take $\mu\equiv 0$ in (\ref{ch3.stochmodel}) then equivalently in (\ref{ch3.pide.forward.eq}). We leave the end of the proof to the reader.
\end{proof}
\subsection{Markovian jump-diffusion models}
Another important particular case in the literature is the case of a Markov jump-diffusion  driven by a Poisson random measure.   Andersen and Andreasen \cite{andersen} derived a forward PIDE in the situation where the jumps are driven by a compound Poisson process  with time-homogeneous Gaussian jumps.
We will now show here that Theorem \ref{ch3.pide.forward.prop} implies the PIDE derived  in \cite{andersen}, given here in a more general context
allowing for a time- and state-dependent L\'evy measure, as well as infinite number of jumps per unit time (``infinite jump activity").
\begin{proposition}[Forward PIDE for jump diffusion model]
Consider the price process $S$ whose dynamics under the pricing measure
$\mathbb{P}$ is given by
\begin{equation}
S_t = S_0 +\int_0^T r(t)S_{t-} dt + \int_0^T S_{t-}\sigma(t,S_{t-}) dB_t + \int_0^T\int_{-\infty}^{+\infty} S_{t-}(e^y - 1) \tilde{N}(dt dy)
\end{equation}
where $B_t$ is a Brownian motion and $N$ a Poisson random measure on $[0,T]\times\mathbb{R}$ with compensator
$\nu(dz)\,dt$, $\tilde{N}$ the associated compensated random measure.
Assume that
   \begin{equation}\label{ch3.A'1a}
      \sigma(.,.)\quad{\rm is\ bounded}\qquad{\rm and}\quad \int_{ \{|y|>1\}} e^{2y} \nu(dy) <\infty.
   \end{equation}
Then the call option price
$$
C_{t_0}(T,K)=e^{-\int_{t_0}^T r(t)\,dt}E^{\mathbb{P}}[\max(S_T-K,0)|\mathcal{F}_{t_0}]
$$
 is a solution (in the sense of distributions) of the PIDE
\begin{equation}\label{ch3.pide.andersen}
\begin{split}
  \frac{\partial C_{t_0}}{\partial
T}(T,K) &= -r(T)K\frac{\partial C_{t_0}}{\partial K}(T,K)+\frac{K^2\sigma(T,K)^2}{2}\, \frac{\partial^2 C_{t_0}}{\partial
K^2}(T,K) \\
&\quad\quad +\int_{\mathbb{R}} \nu(dz)\,e^z\left[C_{t_0}(T,Ke^{-z})-C_{t_0}(T,K)-K(e^{-z}-1)\frac{\partial C_{t_0}}{\partial K}\right]
\end{split}
\end{equation}
on $[t_0,\infty[\times ]0,\infty[$ with
the initial condition:
$$\forall K>0,\quad C_{t_0}(t_0,K)= (S_{t_0}-K)_+.$$
\end{proposition}
\begin{proof}
%The purpose of the proof is to show how one could recover
%(\ref{ch3.pide.andersen}) from (\ref{ch3.pide.forward.eq})
%(actually, one could observe that there exists a direct proof of
%(\ref{ch3.pide.andersen}), for example by Tanaka (see \cite[Proposition 12.2]{contankov}).
As in the proof of Theorem \ref{ch3.pide.forward.prop}, by replacing $\mathbb{P}$ by the conditional measure ${\mathbb{P}}_{\mathcal{F}_{t_0}}$  given $\mathcal{F}_{t_0}$, we may
replace the conditional expectation in (\ref{ch3.def.call}) by an expectation with
respect to the marginal distribution $p^S_T(dy)$ of $S_T$ under ${\mathbb{P}}_{|\mathcal{F}_{t_0}}$.
Thus, without loss of generality, we put $t_0=0$ in the sequel, consider the case where ${\mathcal{F}_0}$ is the $\sigma$-algebra generated by all $\mathbb{P}$-null sets and we denote $C_0(T,K)\equiv C(T,K)$ for simplicity.\\
Differentiating (\ref{ch3.def.call}) in the sense of distributions with respect to $K$, we obtain:
\begin{eqnarray*}
&&\frac{\partial C}{\partial K}(T,K)=-e^{-\int_0^T r(t)\,dt}\int_K^\infty p^S_{T}(dy),\qquad\frac{\partial^2 C}{\partial K^2}(T,dy)=e^{-\int_0^T r(t)\,dt}p^S_{T}(dy).
\end{eqnarray*}
In this particular case, $m(t,dz)\,dt\equiv\nu(dz)\,dt$ and $\psi_t$ are simply given by:
\begin{equation}\
  \psi_{t}(z)\equiv\psi(z)=
  \begin{cases}
&\int_{-\infty}^z dx\  e^x \int_{-\infty}^x \nu(du) \quad z<0\nonumber\\
&\int_{z}^{+\infty} dx\  e^x \int_x^{\infty} \nu(du)\quad z>0\nonumber
\end{cases}
\end{equation}
Then \eqref{ch3.new.para} yields
$$
\chi_{t,S_{t-}}(z)=\mathbb{E}\left[\psi_t\left(z\right)|S_{t-}\right]=\psi(z).
$$ %and $\psi$ depends only on the state space variable.\\
Let us now focus on the term
$$ \int_{0}^{+\infty} y\frac{\partial^2 C}{\partial K^2}(T,dy)\,\chi\left(\ln{\left(\frac{K}{y}\right)}\right)$$
in (\ref{ch3.pide.forward.eq}). Applying Lemma \ref{ch3.calcul.de.B} yields
\begin{eqnarray}
&&\int_{0}^{+\infty} y\frac{\partial^2 C}{\partial K^2}(T,dy)\,\chi\left(\ln{\left(\frac{K}{y}\right)}\right)\nonumber\\
&=&\int_0^\infty e^{-\int_0^T r(t)\,dt}p^S_{T}(dy)\int_{\mathbb{R}}[(ye^z-K)^+-e^z(y-K)^+-K(e^z-1)1_{ \{y>K\}}]\nu(dz)\nonumber\\
&=& \int_{\mathbb{R}}e^z\int_0^\infty e^{-\int_0^T r(t)\,dt}\,p^S_{T}(dy)[(y-Ke^{-z})^+-(y-K)^+-K(1-e^{-z})1_{ \{y>K\}}]\,\nu(dz)\nonumber\\
&=&\int_{\mathbb{R}} e^z\left[C(T,Ke^{-z})-C(T,K)-K(e^{-z}-1)\frac{\partial C}{\partial K}\right]\,\nu(dz).
\end{eqnarray}
This ends the proof.
\end{proof}

\subsection{Pure jump processes}\label{ch3.purejump.sec}
For price processes  with no Brownian component,
Assumption \eqref{ch3.H}  reduces to
$$
  \forall T>0, \quad \mathbb{E}\left[\exp{\left(\int_0^T dt \int (e^y-1)^2 m(t,\,dy)\right)}\right]<\infty.
$$
Assume there exists a measurable function $\chi:[t_0,T]\times \mathbb{R}^+-\{0\}\mapsto \mathbb{R}^+$ such that for all $t\in[t_0,T]$ and for all $z\in\mathbb{R}$:
\begin{equation}
 \chi_{t,S_{t-}}(z)=\mathbb{E}\left[\psi_t\left(z\right)|S_{t-}\right],
\end{equation}
with
$$
\psi_{T}(z)=
 \begin{cases}
&\int_{-\infty}^z dx\  e^x \int_{-\infty}^x m(T,du),\quad z<0\ ;\\
&\int_{z}^{+\infty} dx\  e^x \int_x^{\infty} m(T,du), \quad z>0,\\
\end{cases}
$$
then, the forward equation for call option becomes
\begin{equation}
  \frac{\partial C}{\partial T} + r(T)K\frac{\partial C}{\partial K} = \int_{0}^{+\infty} y\frac{\partial^2 C}{\partial K^2}(T,dy)\,\chi_{T,y}\left(\ln{\left(\frac{K}{y}\right)}\right).
\end{equation}
It is convenient to use the change of variable: $v=\ln y, k=\ln K$. Define $c(k,T)=C(e^k,T)$. Then one can write this PIDE as
\begin{equation} \label{ch3.purejump.eq}
  \frac{\partial c}{\partial
T} + r(T)\frac{\partial c}{\partial k} = \int_{-\infty}^{+\infty} e^{2(v-k)}\left(\frac{\partial^2 c}{\partial k^2}-\frac{\partial c}{\partial k}\right)(T,dv)\,\chi_{T,v}(k-v).
\end{equation}
%where $\chi_{T,v}$ is defined via:
%$$
%  \chi_{T, \log(S_{t-})}(z)=\mathbb{E}\left[\psi_{T}(z)|S_{T-}\right].
%$$
In the case, considered in \cite{carrgeman04}, where the L\'evy density $m_Y$ has a deterministic separable form
\begin{equation}\label{ch3.separable.sec}
m_Y(t,dz,y)\,dt=\alpha(y,t)\,k(z)\,dz\,dt,
\end{equation}
 Equation (\ref{ch3.purejump.eq}) allows us to recover\footnote{Note however that
the equation given in \cite{carrgeman04} does not seem to be correct: it involves the double tail of $k(z)\,dz$ instead of the exponential double tail.} equation (14) in \cite{carrgeman04}
$$
  \frac{\partial c}{\partial
T} + r(T)\frac{\partial c}{\partial k} = \int_{-\infty}^{+\infty} \kappa(k-v)e^{2(v-k)}\alpha(e^{v},T)\left(\frac{\partial^2 c}{\partial k^2}-\frac{\partial c}{\partial k}\right)(T,dv)\,
$$
where $\kappa$ is defined as the exponential double tail of $k(u)\,du$, i.e.
$$
\kappa(z)=
 \begin{cases}
&\int_{-\infty}^z dx\  e^x \int_{-\infty}^x k(u)\,du\quad z<0\ ;\\
&\int_{z}^{+\infty} dx\  e^x \int_x^{\infty} k(u)\,du \quad z>0.\\
\end{cases}
$$
The right hand side can be written as a  convolution of distributions:
\begin{eqnarray}\label{ch3.locallevy.eq}\frac{\partial c}{\partial
T} + r(T)\frac{\partial c}{\partial k} = [a_T(.)\ \left(\frac{\partial^2 c}{\partial k^2}-\frac{\partial c}{\partial k}\right)]*g\qquad
 {\rm where}\\
g(u)= e^{-2u}\kappa(u)\qquad a_T(u)=\alpha(e^{u},T).
\end{eqnarray}
Therefore, knowing $c(.,.)$ and given $\kappa(.)$ we can recover $a_T$ hence $\alpha(.,.)$.
As noted by Carr et al. \cite{carrgeman04}, this equation is analogous to the Dupire formula for diffusions: it enables to ``invert" the structure of the jumps--represented by $\alpha$-- from the cross-section of option prices.
Note that, like the Dupire formula, this inversion involves a double deconvolution/differentiation of $c$ which illustrates the ill-posedness of the inverse problem.
\subsection{Time changed  L\'evy processes}
Time changed L\'evy processes were proposed in \cite{carrgeman03} in the context of option pricing.
Consider the price process $S$ whose dynamics under the pricing measure
$\mathbb{P}$ is given by
\begin{equation}
S_t\equiv e^{\int_0^tr(u)\,du}\,X_t\qquad X_t= \exp{\left(L_{\Theta_t}\right)} \qquad \Theta_t=\int_0^t \theta_s ds
\end{equation}
where
 $L_t$ is a L\'evy process with characteristic
triplet $(b,\sigma^2,\nu)$, $N$ its jump measure and $(\theta_t)$ is a locally bounded
 positive semimartingale. %We assume $L$ and $\theta$ are $\mathcal{F}_t$-adapted.\\
 $X$ is a $\mathbb{P}$-martingale if %$(\mathbb{P},\mathcal{F}^L_t)$ (see \cite[Lemma 15.]{contankov})
 %which requires
\begin{equation}
 b+\frac{1}{2}\sigma^2+\int_{\mathbb{R}}(e^z-1-z\,1_{ \{|z|\leq 1\}})\nu(dy)=0.
\end{equation}
Define the value $C_{t_0}(T,K)$ at  $t_0$ of the call option with expiry $T>t_{0}$
and strike $K>0$ as
\begin{equation}
C_{t_0}(T,K)=e^{-\int_0^Tr(t)\,dt}E^\mathbb{P}[\max(S_T-K,0)|\mathcal{F}_{t_0}].%=E^\mathbb{P}[\max(X_T-Ke^{-\int_0^Tr(t)\,dt},0)|\mathcal{G}_{t_0}]
%=E^\mathbb{P}[\max(X_T-K,0)|\mathcal{F}_{t_0}]
\end{equation}
%Assume $(A_{2a}')$
%\begin{equation}
%   \begin{cases}
%     %&\int(1\wedge y^2)\,n(y)\,dy <\infty\quad a.s.\quad (A_{1c})\\[0.1cm]
%     & \forall T>0\: \int_0^T \int_{ y>1} e^{2y} \,n(y)\,dy\,dt <\infty\quad a.s.\quad(A_{2c})
%     \end{cases}
% \end{equation}
\begin{proposition}\label{ch3.pide.forward.levy.th}
Assume there exists a measurable function $\alpha:[0,T]\times \mathbb{R}\mapsto \mathbb{R}$
such that
\begin{equation}
 \alpha(t,X_{t-})=E[\theta_t| X_{t-}],
\end{equation}
and let $\chi$ be the exponential double tail of $\nu$, defined as
\begin{equation}
  \chi(z)=
  \begin{cases}
&\int_{-\infty}^z dx\  e^x \int_{-\infty}^x \nu(du), \quad z<0\ ;\\
&\int_{z}^{+\infty} dx\  e^x \int_x^{\infty} \nu(du), \quad z>0.\\
\end{cases}
\end{equation}
If  $\beta=\frac{1}{2}\sigma^2 + \int_{\mathbb{R}} (e^y-1)^2 \nu(dy)<\infty$ and
\begin{equation}
  \mathbb{E}\left[\exp{(\beta \Theta_T )}\right]<\infty, \label{ch3.integrability.timechange}
\end{equation}
then the call
option price $C_{t_0}:(T,K)\mapsto C_{t_0}(T,K)$ at date $t_0$, as a function of maturity and
 strike, is a solution (in the sense of distributions) of the partial
integro-differential equation
\begin{equation}\label{ch3.pide.forward.levy.eq}
\begin{split}
  \frac{\partial C}{\partial
T}(T,K) &= -r\alpha(T,K)K\frac{\partial C}{\partial K}(T,K)+\frac{K^2\alpha(T,K)\sigma^2}{2}\, \frac{\partial^2 C}{\partial
K^2}(T,K) \\
&+\int_{0}^{+\infty} y\frac{\partial^2 C}{\partial K^2}(T,dy)\,\alpha(T,y)\,\chi\left(\ln{\left(\frac{K}{y}\right)}\right)
\end{split}
\end{equation}
%r(T)C(T,K) ?? histoire de martingale c'est s  r
on $[t,\infty[\times ]0,\infty[$ with
the initial condition:\ $\forall K>0\quad C_{t_0}(t_0,K)= (S_{t_0}-K)_+.$
\end{proposition}
\begin{proof}
Using Lemma \cite[Lemma 2]{bentatacont09}, $(L_{\Theta_t})$ writes
  \begin{eqnarray}
    L_{\Theta_t} &=& L_0+\int_0^t \sigma  \sqrt{\theta_s} dB_s + \int_0^t b \theta_s ds \nonumber \\[0.1cm]
    &+& \int_0^t\theta_s\int_{|z|\leq 1} z \tilde{N}(ds\,dz) + \int_0^t\int_{\{|z|>1\}} z N(ds\,dz)\nonumber
 \end{eqnarray}
where $N$ is an integer-valued random measure with compensator
$\theta_t\nu(dz)\,dt$, $\tilde{N}$ its compensated random measure.
Applying the It\^{o} formula yields
\begin{eqnarray}
X_t&=&X_0+\int_0^t X_{s-}dL_{T_s}+\frac{1}{2}\int_0^t X_{s-}\sigma^2\theta_s\,ds
+\sum_{s\leq t} \left(X_s-X_{s-}-X_{s-}\Delta L_{T_s}\right)\nonumber\\
&=&X_0+\int_0^t X_{s-}\left[b\theta_s+\frac{1}{2}\sigma^2\theta_s\right]\,ds+\int_0^tX_{s-}\sigma\sqrt{\theta_s}\,dB_s\nonumber\\
&+& \int_0^tX_{s-}\theta_s\int_{\{|z|\leq 1\}} z \tilde{N}(ds\,dz) + \int_0^tX_{s-}\theta_s\int_{\{|z|>1\}} z N(ds\,dz)\nonumber\\
&+& \int_0^t\int_{\mathbb{R}} X_{s-}(e^z-1-z)N(ds\,dz)\nonumber
\end{eqnarray}
Under our assumptions, $\int(e^z-1-z\,1_{ \{|z|\leq 1\}})\nu(dz)<\infty$, hence:
\begin{eqnarray}
X_t&=&X_0+\int_0^t X_{s-}\left[b\theta_s+\frac{1}{2}\sigma^2\theta_s+\int_{\mathbb{R}}(e^z-1-z\,1_{\{ |z|\leq 1\}})\theta_s\nu(dz)\right]\,ds+\int_0^tX_{s-}\sigma\sqrt{\theta}\,dB_s\nonumber\\
&+& \int_0^t\int_{\mathbb{R}} X_{s-}\theta_s(e^z-1)\tilde{N}(ds\,dz)\nonumber\\
&=&X_0 + \int_0^tX_{s-}\sigma\sqrt{\theta_s}\,dB_s + \int_0^t\int_{\mathbb{R}} X_{s-}(e^z-1)\tilde{N}(ds\,dz)\nonumber
\end{eqnarray}
and $(S_t)$ may be expressed as
$$
S_t=S_0 + \int_0^t S_{s-} r(s)\,ds+ \int_0^tS_{s-}\sigma\sqrt{\theta_s}\,dB_s + \int_0^t\int_{\mathbb{R}} S_{s-}(e^z-1)\tilde{N}(ds\,dz).
$$
Assumption \eqref{ch3.integrability.timechange} implies that $S$ fulfills Assumption $(H)$ of Theorem \ref{ch3.pide.forward.prop} and $(S_t)$ is now in the suitable form (\ref{ch3.stochmodel}) to apply Theorem \ref{ch3.pide.forward.prop}, which yields the result.
%$$
%    \begin{split}
%      \beta^Y(t,z)&=\mathbb{E}\left[\beta_t^X|X_{t^-}=z\right]=b\,\alpha(t,z);\\
%      \delta^Y(t,z)&=\mathbb{E}\left[\delta_t^X|X_{t^-}=z\right]=\sigma\,\sqrt{\alpha(t,z)};\\
%      m_Y(t,y,z)&=\mathbb{E}\left[m_X(t,y)|X_{t^-}=z\right]=\alpha(t,z)\nu_L(dy).
%    \end{split}
%  $$
\end{proof}

\subsection{Index options in a multivariate jump-diffusion model}\label{ch3.index.sec}
Consider  a multivariate model with $d$ assets
$$
S_T^i = S_0^i +\int_0^T r(t) S_{t^-}^i dt + \int_0^T S_{t^-}\delta_t^i dW_t^i + \int_0^T\int_{\mathbb{R}^d} S_{t^-}^i(e^{y_i} - 1) \tilde{N}(dt\,dy)
$$
where $\delta^i$ is an adapted process taking values in $\mathbb{R}$ representing the volatility of asset $i$, $W$ is a d-dimensional Wiener process, $N$ is a Poisson random measure on $[0,T]\times\mathbb{R}^d$ with compensator $\nu(dy)\,dt$, $\tilde{N}$ denotes its compensated random measure. The Wiener processes $W^i$ are correlated
$$\forall 1\leq (i,j)\leq d, \langle W^i, W^j \rangle_t=\rho_{i,j}t,$$
with $\rho_{ij}>0$ and $\rho_{ii}=1$. An index is defined as a weighted sum of asset prices
$$
  I_t=\sum_{i=1}^d w_i S^i_t \quad w_i>0\quad \sum_1^d w_i=1.
$$
The value $C_{t_0}(T,K)$ at time $t_0$ of an index call option with expiry $T>t_{0}$
and strike $K>0$ is given by
\begin{equation}
C_{t_0}(T,K)=e^{-\int_{t_0}^T r(t)\,dt}E^{\mathbb{P}}[\max(I_T-K,0)|\mathcal{F}_{t_0}].
\end{equation}
The following result is a generalization of the forward PIDE studied by Avellaneda et al. \cite{avellaneda03} for the diffusion case:
\begin{theorem}{Forward PIDE for index  options.}\label{ch3.pide.index}
Assume
\begin{equation}\label{ch3.assum.tris}
   \begin{cases}
     &\forall T>0 \quad \mathbb{E}\left[\exp{\left(\frac{1}{2}\int_0^T\|\delta_t\|^2\,dt\right)}\right]<\infty \\[0.1cm]
     &\int_{\mathbb{R}^d}(1\wedge \|y\|)\,\nu(dy)<\infty\quad a.s.\\[0.1cm]
     & \int_{ \{\|y\|>1\}} e^{2\|y\|} \nu(dy)<\infty\quad a.s.
   \end{cases}
 \end{equation}
Define
\begin{equation}\label{ch3.def.eta}
\eta_t(z)=
\begin{cases}
&\int_{-\infty}^z dx\  e^x \int_{ \mathbb{R}^d} 1_{\ln{\left(\frac{\sum_{1\leq i\leq d-1} w_iS_{t-}^i e^{y_i}}{I_{t-}}\right)}\leq x }\nu(dy)\quad z<0\\
&\int_{z}^\infty dx\  e^x \int_{ \mathbb{R}^d}
1_{\ln{\left(\frac{\sum_{1\leq i\leq d-1} w_iS_{t-}^i
e^{y_i}}{I_{t-}}\right)}\geq x }\nu(dy) \quad z>0
\end{cases}
\end{equation}
and assume there exists measurable functions $\sigma:[t_0,T]\times \mathbb{R}^+-\{0\}\mapsto \mathbb{R}^+$, $\chi:[t_0,T]\times \mathbb{R}^+-\{0\}\mapsto \mathbb{R}^+$ such that for all $t\in[t_0,T]$ and for all $z\in\mathbb{R}$:
%and define, for $t> t_0, z>0$,
\begin{equation}
\begin{cases}
    \sigma(t,I_{t-})&=\frac{1}{z}\,\sqrt{\mathbb{E}\left[\left(\sum_{i,j=1}^d w_iw_j\rho_{ij}\,\delta_t^i\delta_t^j\,S^{i}_{t-} S^j_{t-}\right)|I_{t^-}\right]}\quad a.s.,\\
\chi_{t,I_{t-}}(z)&=\mathbb{E}\left[\eta_t\left(z\right)|I_{t-}\right]\quad a.s.
\end{cases}
\end{equation}
Then the index call  price $(T,K)\mapsto C_{t_0}(T,K)$, as a function of maturity and
strike, is a solution (in the sense of distributions) of the partial
integro-differential equation
\begin{equation} \label{ch3.pide.forward.eq.index}
  \frac{\partial C_{t_0}}{\partial
T} = -r(T)K\frac{\partial C{t_0}}{\partial K}+\frac{\sigma(T,K)^2}{2}\, \frac{\partial^2 C{t_0}}{\partial
K^2}+\int_{0}^{+\infty} y\frac{\partial^2 C{t_0}}{\partial K^2}(T,dy)\,\chi_{T,y}\left(\ln{\left(\frac{K}{y}\right)}\right)
\end{equation}
on $[t_0,\infty[\times ]0,\infty[$ with
the initial condition:
$$\forall K>0,\quad C_{t_0}(t_0,K)= (I_{t_0}-K)_+.$$
%Consider the stochastic differential equation
%\begin{equation}\label{ch3.index.mimic}
%Y_t = S_0 +\int_0^T r(t)Y_{t-} dt + \int_0^T \sigma(t,Y_{t-}) dZ_t +
%\int_0^T\int_{-\infty}^{+\infty} Y_{t-}(e^{y} - 1) \tilde{J}(dt\,dy)
%\end{equation}
%where $(Z_t)$ is a real-valued Brownian motion, $J$ is an integer-valued
%  random measure on $[0,T]\times\mathbb{R}$ with compensator
%  $j(t,y,Y_{t-})\,dt\,dy$ and $\tilde{J}$ the associated compensated random measure.\\

%$((Y_t)_{t\in[0,T]},,\mathbb{Q}_{S_0})$, defined as the solution
%  to the martingale problem on $C^\infty_0(\mathbb{R})$  for the operator
%\begin{equation}
%  \begin{split}
%    Lf(t,x)&=xr(t)\frac{\partial f}{\partial x}(t,x)+\frac{\sigma(t,x)^2}{2}\frac{\partial^2 f}{\partial x^2}(t,x)\\
%    &+\int_{\mathbb{R}^d}[f(t,xe^{y})-f(t,x)-1_{ |y|\leq 1}x(e^{y}-1).\frac{\partial f}{\partial x}f(t,x)]j(t,y,x)\,dy
%  \end{split}
%\end{equation}
%and $(Y_t)_{t\geq 0}$ has the same marginal distributions as $(S_t)_{t\geq 0}$:
%   $$ \forall t\geq 0 \   Y_t\overset{\underset{\mathrm{d}}{}}{=}S_t.$$
\end{theorem}

\begin{proof}
$(B_t)_{t\geq 0}$ defined by
\[
dB_t=\frac{\sum_{i=1}^d w_iS^i_{t-}\delta_t^i dW_t^i}{\left(\sum_{i,j=1}^d w_iw_j\rho_{ij}\,\delta_t^i\delta_t^j\,S^{i}_{t-} S^j_{t-}\right)^{1/2}}
\]
 is a continuous local martingale with quadratic variation $t$: by L\'evy's theorem, $B$ is a Brownian motion.
%One observes that $\sum_{i=1}^d  S^i_{t-}\delta_t^i dW_t^i$ is a martingale with quadratic variation
%$\,\sum_{i,j=1}^d \rho_{ij}\,\delta_t^i\delta_t^j\,S^{i}_{t-} S^j_{t-}$: there exists a real-valued Brownian motion $(B_t)$ such that
%$$
%\sum_{i=1}^d  w_i S^i_{t-}\delta_t^i dW_t^i=\left(\sum_{i,j=1}^d w_iw_j \rho_{ij}\,\delta_t^i\delta_t^j\,S^{i}_{t-} S^j_{t-}\right)^{1/2}dB_t
%$$
Hence $I$ may be decomposed as
\begin{equation}\label{ch3.index}
\begin{split}
I_T &= \sum_{i=1}^d w_i S^i_0+ \int_0^T r(t)I_{t-}\,dt + \int_0^T \left(\sum_{i,j=1}^d w_iw_j\rho_{ij}\,\delta_t^i\delta_t^j\,S^{i}_{t-} S^j_{t-}\right)^{\frac{1}{2}} dB_t\\
& + \int_0^T\int_{\mathbb{R}^d}  \sum_{i=1}^d w_i S^i_{t-}(e^{y_i}-1)\tilde{N}(dt\,dy)
\end{split}
\end{equation}
The essential part of the proof consists in rewriting $(I_t)$ in the suitable form (\ref{ch3.stochmodel}) to apply Theorem \ref{ch3.pide.forward.prop}.
%(\ref{ch3.index})may be expressed as:
%\begin{eqnarray}
%dI_t &=& r(t)I_{t-} \,dt + \frac{\left(\sum_{i,j=1}^d w_iw_j\rho_{ij}\,\delta_t^i\delta_t^j\,S^{i}_{t-} S^j_{t-}\right)^{\frac{1}{2}}}{I_{t-}}I_{t-}\,dB_t \\
%&+& \int_{\mathbb{R}^d}  \frac{\sum_{i=1}^d w_i S^i_{t-}(e^{y_i}-1)}{I_{t-}}I_{t-}\tilde{N}(dt\,dy)\\
%&=&r(t)I_{t-} \,dt + \frac{\left(\sum_{i,j=1}^d w_iw_j\rho_{ij}\,\delta_t^i\delta_t^j\,S^{i}_{t-} S^j_{t-}\right)^{\frac{1}{2}}}{I_{t-}}I_{t-}\,dB_t \\
%&+& \int_{\mathbb{R}^d}  \left(\frac{\sum_{1\leq i\leq d} w_iS_{t-}^i e^{y_i}}{I_{t-}}-1\right)I_{t-}\tilde{N}(dt\,dy)
%\end{eqnarray}
Applying the It\^{o} formula to $\ln{(I_T)}$ yields
\begin{eqnarray*}
   \ln \frac{I_T}{I_0}
&=&\int_0^T \Big[r(t)-\frac{1}{2I_{t-}^2}\sum_{i,j=1}^d w_iw_j\rho_{ij}\,\delta_t^i\delta_t^j\,S^{i}_{t-} S^j_{t-}\nonumber\\
&& -\int\big(\frac{\sum_{1\leq i\leq d} w_iS_{t-}^i e^{y_i}}{I_{t-}}-1-\ln{\big(\frac{\sum_{1\leq i\leq d} w_iS_{t-}^i e^{y_i}}{I_{t-}}\big)}\big)\nu(dy)\Big]\,dt\nonumber\\
&+& \int_0^T\frac{1}{I_{t-}}\left(\sum_{i,j=1}^d w_iw_j\rho_{ij}\,\delta_t^i\delta_t^j\,S^{i}_{t-} S^j_{t-}\right)^{\frac{1}{2}}\,dB_t
  +\int_0^T\int \ln{\left(\frac{\sum_{1\leq i\leq d} w_iS_{t-}^i e^{y_i}}{I_{t-}}\right)}\,\tilde{N}(dt\:dy).
\end{eqnarray*}
Using the convexity property of the logarithm,
\begin{equation*}
\ln{\left(\frac{\sum_{1\leq i\leq d} w_iS_{t-}^i e^{y_i}}{I_{t-}}\right)}\geq \frac{\sum_{1\leq i\leq d} w_iS_{t-}^i y_i}{I_{t-}}\geq -\|y\|,\quad {\rm and}\quad
\ln{\left(\frac{\sum_{1\leq i\leq d} w_iS_{t-}^i e^{y_i}}{I_{t-}}\right)}\leq \ln{\left( \max_{1\leq i \leq d}\, e^{y_i}\right)},
\end{equation*}
implying that
\begin{equation*}
\left|\ln{\left(\frac{\sum_{1\leq i\leq d} w_iS_{t-}^i e^{y_i}}{I_{t-}}\right)}\right|\leq\left|\sum_{1\leq i\leq d} \frac{w_iS_{t-}^i}{I_{t-}} y_i\right|
\leq\sum_{1\leq i\leq d} |y_i|\leq \|y\|,
\end{equation*}
so the functions $y\to\ln{\left(\frac{\sum_{1\leq i\leq d} w_iS_{t-}^i e^{y_i}}{I_{t-}}\right)}$ and $y\to \frac{\sum_{1\leq i\leq d} w_iS_{t-}^i e^{y_i}}{I_{t-}}$
are integrable with respect to $\nu(dy)$ under the assumptions \eqref{ch3.assum.tris}.
We furthermore observe that
\begin{equation}\label{ch3.holder}
 \begin{split}
  &\int  1\wedge \left|\ln{\left(\frac{\sum_{1\leq i\leq d} w_iS_{t-}^i e^{y_i}}{I_{t-}}\right)}\right|\,\nu(dy)<\infty\quad a.s.\\
  & \int_0^T \int_{ \{\|y\|>1\}} e^{2\left|\ln{\left(\frac{\sum_{1\leq i\leq d} w_iS_{t-}^i e^{y_i}}{I_{t-}}\right)}\right|} \nu(dy)\,dt<\infty\quad a.s.
   \end{split}
 \end{equation}
Similarly,  \eqref{ch3.assum.tris} implies that  $\int\left(e^{y_i}-1-1_{\{ |y_i|\leq 1\}}y_i\,\right)\nu(dy)<\infty$ so $\ln{(S_T^i)}$ may be expressed as
\begin{eqnarray*}
  \ln{(S_T^i)}&=&\ln{(S_0^i)}+\int_0^T \big( r(t)-\frac{1}{2}(\delta_t^i)^2-\int \left(e^{y_i}-1-1_{ \{|y_i|\leq 1\}}y_i\,\right)\nu(dy)\big)\,dt\nonumber\\
&+&\int_0^T \delta_t^i\,dW_t^i +\int_0^T\int y_i\,\tilde{N}(dt\:dy)\nonumber
\end{eqnarray*}
%%%%%%%%%%%%%%%%%%%%%%%%%
%le nouveau brownien
Define the d-dimensional martingale $W_t=(W_t^1,\cdots,W_t^{d-1},B_t)$.
For $1\leq i,j \leq d-1$ we have
\[
\langle W^i,W^j \rangle_t=\rho_{i,j}t \quad{\rm and}\quad \langle W^i, B\rangle_t=\frac{\sum_{j=1}^d w_j\rho_{ij}S_{t-}^j\delta_t^j}{\left(\sum_{i,j=1}^d w_iw_j \rho_{ij}\,\delta_t^i\delta_t^j\,S^{i}_{t-} S^j_{t-}\right)^{1/2}}\,t.
\]
Define
\begin{small}
\[\Theta_t=
\left(\begin{array}{cccc}
 1& \cdots & \rho_{1,d-1} & \frac{\sum_{j=1}^d w_j\rho_{1j}S_{t-}^j\delta_t^j}{\left(\sum_{i,j=1}^d w_iw_j \rho_{ij}\,\delta_t^i\delta_t^j\,S^{i}_{t-} S^j_{t-}\right)^{1/2}}\\
\vdots &\ddots & \vdots &\vdots\\
\rho_{d-1,1}&\cdots &1 & \frac{\sum_{j=1}^d w_j\rho_{d-1,j}S_{t-}^j\delta_t^j}{\left(\sum_{i,j=1}^d w_iw_j \rho_{ij}\,\delta_t^i\delta_t^j\,S^{i}_{t-} S^j_{t-}\right)^{1/2}}\\
\frac{\sum_{j=1}^d w_j\rho_{1,j}S_{t-}^j\delta_t^j}{\left(\sum_{i,j=1}^d w_iw_j \rho_{ij}\,\delta_t^i\delta_t^j\,S^{i}_{t-} S^j_{t-}\right)^{1/2}}
&\cdots &\frac{\sum_{j=1}^d w_j\rho_{d-1,j}S_{t-}^j\delta_t^j}{\left(\sum_{i,j=1}^d w_iw_j \rho_{ij}\,\delta_t^i\delta_t^j\,S^{i}_{t-} S^j_{t-}\right)^{1/2}}
&1
\end{array}\right)
\]
\end{small}
There exists a standard Brownian motion $(Z_t)$ such that $W_t=A Z_t$ where $A$ is a $d\times d$ matrix verifying $\Theta={}^tA\,A$.
Define $X_T\equiv \left(\ln{(S_T^1)},\cdots, \ln{(S_T^{d-1})},\ln{(I_T)}\right)$;
\begin{equation*}
\delta=
\left(\begin{array}{cccc}
\delta_t^1 & \cdots & 0 & 0\\
\vdots & \ddots & \vdots & \vdots\\
 0 & \cdots & \delta_t^{d-1}& 0\\
0 & \cdots & 0 & \frac{1}{I_{t-}}\left(\sum_{i,j=1}^d w_iw_j\rho_{ij}\,\delta_t^i\delta_t^j\,S^{i}_{t-} S^j_{t-}\right)^{\frac{1}{2}}
\end{array}\right),
\end{equation*}
\begin{small}
\[\beta_t=
\left(\begin{array}{c}
r(t)-\frac{1}{2}(\delta_t^1)^2-\int  \left(e^{y_1}-1-y_1\right)\,\nu(dy)\\
  \vdots\\
r(t)-\frac{1}{2}(\delta_t^{d-1})^2-\int  \left(e^{y_{d-1}}-1-y_{d-1}\right)\,\nu(dy)\\
r(t)-\frac{1}{2I_{t-}^2}\sum_{i,j=1}^d w_iw_j\rho_{ij}\,\delta_t^i\delta_t^j\,S^{i}_{t-} S^j_{t-}-\int
\left(\frac{\sum_{1\leq i\leq d} w_iS_{t-}^i e^{y_i}}{I_{t-}}-1-\ln{\left(\frac{\sum_{1\leq i\leq d} w_iS_{t-}^i e^{y_i}}{I_{t-}}\right)}\right)\,\nu(dy)
\end{array}\right),
\]
\end{small}
%h_t(y)&=&\left(e^{y_1}-1-y_1, \cdots, e^{y_{d-1}}-1-y_{d-1}, e^{\tilde{\psi}_t(y)}-1-\tilde{\psi}_t(y),\right)\\
\[
{\rm and}\qquad \psi_t(y)=
\left(\begin{array}{c}
  y_1\\
  \vdots\\
  y_{d-1}\\
  \ln{\left(\frac{\sum_{1\leq i\leq d} w_iS_{t-}^i e^{y_i}}{I_{t-}}\right)}
\end{array}\right).
\]
Then $X_T$ may be expressed as
\begin{equation}
  X_T=X_0+\int_0^T \beta_t\,dt+\int_0^T \delta_tA\,dZ_t +\int_0^T\int_{\mathbb{R}^d} \psi_t(y)\,\tilde{N}(dt\:dy)
\end{equation}
%$\left(\bar{\delta}_t\right)_{1\leq i,j\leq d}$ is the $M_{d\times d}(\mathbb{R})$-matrix defined by:\\
%\begin{tiny}
%\[\bar{\delta}_t=
%\left(\begin{array}{cccc}
%   \left(\delta_t^1\right)^2 & \cdots & \delta_t^1\delta_t^{d-1} & \frac{\delta_t^1\left(\sum_{i,j=1}^d w_iw_j\rho_{ij}\,\delta_t^i\delta_t^j\,S^{i}_{t-} S^j_{t-}\right)^{\frac{1}{2}}}{I_{t-}} \\
%   \vdots &  \ddots &\vdots & \vdots\\
%  \delta_t^{d-1}\delta_t^1 &  \cdots & \left(\delta_t^{d-1}\right)^2& \frac{\delta_t^{d-1}\left(\sum_{i,j=1}^d w_iw_j\rho_{ij}\,\delta_t^i\delta_t^j\,S^{i}_{t-} S^j_{t-}\right)^{\frac{1}{2}}}{I_{t-}} \\
%   \frac{\left(\sum_{i,j=1}^d w_iw_j\rho_{ij}\,\delta_t^i\delta_t^j\,S^{i}_{t-} S^j_{t-}\right)^{\frac{1}{2}}\delta_t^1}{I_{t-}} &
%\cdots &
%\frac{\left(\sum_{i,j=1}^d w_iw_j\rho_{ij}\,\delta_t^i\delta_t^j\,S^{i}_{t-} S^j_{t-}\right)^{\frac{1}{2}} \delta_t^{d-1}}{I_{t-}} &
%\frac{\left(\sum_{i,j=1}^d w_iw_j\rho_{ij}\,\delta_t^i\delta_t^j\,S^{i}_{t-} S^j_{t-}\right)}{I_{t-}}
%\end{array}\right)
%\]
%\end{tiny}
%for $1\leq i,j\leq (d-1)$,
%$\bar{\delta}_t^{ij}=\left(\delta_t^i\delta_t^j\right)^{1/2}$; for $1\leq i \leq (d-1)$,
%$\bar{\delta}_t^{i\,d}=\bar{\delta}_t^{d\,i}\left(\delta_t^i\tilde{\delta}_t\right)^{1/2}$; and
%$\bar{\delta}_t^{d\,d}=\tilde{\delta}_t$; and
%\left(y_1,\cdots,y_{d-1},\tilde{\psi}_t(y)\right)
The predictable function  $\phi_t$ defined, for $t \in [0,T], y\in \psi_t(\mathbb{R}^d)$, by
$$
\phi_t(y)=\left(y_1,\cdots,y_{d-1},\ln{\left(\frac{e^{y_d}I_{t-}-\sum_{1\leq i\leq d-1}w_i S_{t-}^i e^{y_i}}{w_d S_{t-}^d}\right)}\right)
%&=&\ln{\left(e^{y_d}+\frac{\sum_{1\leq i\leq d-1} w_i S_{t-}^i\left(e^{y_d}-e^{y_i}\right)}{w_dS_{t-}^d}\right)}
$$
  is the left inverse   of $\psi_t$: $\phi_t(\omega,\psi_t(\omega, y) )= y.$
Observe that $\psi_t(.,0)=0$, $\phi$ is predictable, and $\phi_t(\omega, .)$ is differentiable on $Im(\psi_t)$ with
Jacobian matrix $\nabla_y\phi_t(y)$ given by
\begin{small}
\begin{equation*}
(\nabla_y\phi_t(y))=\left(
\begin{array}{cccc}
1 & 0 &0 & 0\\
\vdots & \ddots & \vdots& \vdots\\
0& \cdots & 1 & 0\\
\frac{-e^{y_1}w_1S_{t-}^1}{e^{y_d}I_{t-}-\sum_{1\leq i\leq d-1}w_i S_{t-}^i e^{y_i}}  & \cdots & \frac{-e^{y_{d-1}}w_{d-1}S_{t-}^{d-1}}{e^{y_d}I_{t-}-\sum_{1\leq i\leq d-1}w_i S_{t-}^i e^{y_i}}& \frac{e^{y_d}I_{t-}}{e^{y_d}I_{t-}-\sum_{1\leq i\leq d-1}w_i S_{t-}^i e^{y_i}}
\end{array}
\right)
\end{equation*}
\end{small}
%1 \quad\mathrm{for}\:1\leq i=j\leq d-1\\
%\frac{e^{y_j}w_jS_{t-}^j}{e^{y_d}S_{t-}-\sum_{1\leq i\leq d-1}w_i S_{t-}^i e^{y_i}}\quad\mathrm{for}\:i=d, 1\leq j \leq (d-1)\\[0.1cm]
%\frac{e^{y_d}S_{t-}}{e^{y_d}S_{t-}-\sum_{1\leq i\leq d-1}w_i S_{t-}^i e^{y_i}}\quad\mathrm{for}\:i=j=d\\[0.1cm]
%0 \quad\mathrm{else}\\
so $(\psi,\nu)$ satisfies the assumptions of \cite[Lemma 2]{bentatacont09}: using Assumption $(A_{2b})$, for all $T\geq t \geq 0$,
\begin{eqnarray}
&&\mathbb{E}\left[\int_0^T\int_{\mathbb{R}^d} (1\wedge \|\psi_t(.,y)\|^2)\,\nu(dy)\,dt\right]\nonumber\\
&=& \int_0^T\int_{\mathbb{R}^d} 1\wedge \left(y_1^2+\cdots+y_{d-1}^2+\ln{\left(\frac{\sum_{1\leq i\leq d} w_iS_{t-}^i e^{y_i}}{I_{t-}}\right)}^2\right)\,\nu(dy)\,dt \nonumber\\
&\leq& \int_0^T\int_{\mathbb{R}^d} 1\wedge (2\|y\|^2)\,\nu(dy)\,dt
<\infty.\nonumber
\end{eqnarray}
Define $\nu_{\phi}$, the image of $\nu$ by $\phi$ by
\begin{equation}
 \nu_{\phi}(\omega,t,B)=\nu(\phi_t(\omega,B))\quad {\rm for}\quad  B\subset \psi_t(\mathbb{R}^d).
\end{equation}
Applying   \cite[Lemma 2]{bentatacont09}, $X_T$ may be expressed as
$$
X_T= X_0+\int_0^T \beta_t\,dt+ \int_0^T\delta_tA\,dZ_t+\int_0^T\int y\,\tilde{M}(dt\:dy)
$$
where $M$ is an integer-valued random measure (resp. $\tilde{M}$ its compensated random measure) with compensator
$$ \mu(\omega; dt\ dy)=m(t,dy;\omega)\,dt,$$
defined via its density
\begin{eqnarray*}
\frac{d\mu}{d\nu_{\phi}}(\omega,t,y)=1_{\{\psi_t(\mathbb{R}^d)\}}(y)\,|{\rm det}\nabla_y \phi_t|(y)
=1_{\{\psi_t(\mathbb{R}^d)\}}(y)\,\left|\frac{e^{y_d}I_{t-}}{e^{y_d}I_{t-}-\sum_{1\leq i\leq d-1}w_i S_{t-}^i e^{y_i}}\right|
\end{eqnarray*} with respect to $\nu_{\phi}$.
Considering now the d-th component of $X_T$, one obtains the semimartingale decomposition of $\ln{(I_t)}$:
\begin{eqnarray*}
&&\ln{(I_T)}- \ln{(I_0)}\\
&=&\int_0^T \Big(r(t)-\frac{1}{2I_{t-}^2}\left(\sum_{i,j=1}^d w_iw_j\rho_{ij}\,\delta_t^i\delta_t^j\,S^{i}_{t-} S^j_{t-}\right)\nonumber\\
&&\quad\quad\quad-\int \left(\frac{\sum_{1\leq i\leq d} w_iS_{t-}^i e^{y_i}}{I_{t-}}-1-\ln{\left(\frac{\sum_{1\leq i\leq d} w_iS_{t-}^i e^{y_i}}{I_{t-}}\right)}\right)\,\nu(dy)\Big)\,dt\nonumber\\
&+& \int_0^T\frac{1}{I_{t-}} \left(\sum_{i,j=1}^d w_iw_j\rho_{ij}\,\delta_t^i\delta_t^j\,S^{i}_{t-} S^j_{t-}\right)^{\frac{1}{2}}\,dB_t  +\int_0^T\int y\,\tilde{K}(dt\:dy)\nonumber
\end{eqnarray*}
%where
%$$
%\begin{split}
%\tilde{\delta}_t&=\frac{\delta_t}{S_{t-}}\\
%\tilde{\psi}_t(y)&=\ln{\left(\frac{\sum_{1\leq i\leq d} w_iS_{t-}^i e^{y_i}}{S_{t-}}\right)}
%\end{split}
%$$
%\begin{equation}
%\begin{cases}
%\delta_t&= \left(\sum_{i,j=1}^d w_iw_j\rho_{ij}\,\delta_t^i\delta_t^j\,S^{i}_{t-} S^j_{t-}\right)^{\frac{1}{2}}\\
%\psi_t(y)&=\sum_{i=1}^d w_i S^i_{t-}(e^{y_i}-1)
%\end{cases}
%\end{equation}
where $K$ is an integer-valued random measure on
$[0,T]\times\mathbb{R}$ with compensator $k(t,dy)\,dt$ where
\begin{eqnarray*}
k(t,B)&=&  \int_{\mathbb{R}^{d-1}\times B}\mu(t,dy)=\int_{\mathbb{R}^{d-1}\times B} 1_{\{\psi_t(\mathbb{R}^d)\}}(y)\,|{\rm det}\nabla_y \phi_t|(y)\,\nu_{\phi}(t,dy)\\
&=&\int_{\mathbb{R}^{d-1}\times B\cap\psi_t(\mathbb{R}^d)} |{\rm det}\nabla_y \phi_t|(\psi_t(y))\,\nu(dy)\\
&=&\int_{\{y\in\mathbb{R}^d-\{0\}, \ln{\left(\frac{\sum_{1\leq i\leq d-1} w_iS_{t-}^i e^{y_i}}{I_{t-}}\right)}\in B\}}\,\nu(dy)\quad{\rm for}\quad B\in\mathcal{B}(\mathbb{R}-\{0\}).
\end{eqnarray*}
In particular,  the exponential double tail of $k(t,dy)$ which we
denote $\eta_t(z)$
\begin{equation*}
  \eta_t(z)=
  \begin{cases}
&\int_{-\infty}^z dx\  e^x  k(t,]-\infty,x]), \quad z<0\ ;\\
&\int_{z}^{+\infty} dx\  e^x k(t,[x,\infty[),\quad z>0,
\end{cases}
\end{equation*}
is given by \eqref{ch3.def.eta}. So finally $I_T$ may be expressed as
\begin{eqnarray*}
I_T&=&I_0+\int_0^Tr(t)I_{t-} \,dt + \int_0^T\left(\sum_{i,j=1}^d w_iw_j\rho_{ij}\,\delta_t^i\delta_t^j\,S^{i}_{t-} S^j_{t-}\right)^{\frac{1}{2}}\,dB_t\\
&+& \int_0^T\int_{\mathbb{R}^d}  \left(e^{y}-1\right)I_{t-}\tilde{K}(dt\,dy).
\end{eqnarray*}
The normalized volatility of $I_t$ satisfies, for $t\in [0,T]$,
$$
\frac{\sum_{i,j=1}^d w_iw_j \rho_{ij}\,\delta_t^i\delta_t^j\,S^{i}_{t-} S^j_{t-}}{I_{t-}^2}
\leq \sum_{i,j=1}^d \rho_{ij}\,\delta_t^i\delta_t^j,\quad {\rm and}\quad \left|\ln{\left(\frac{\sum_{1\leq i\leq d} w_iS_{t-}^i e^{y_i}}{I_{t-}}\right)}\right|\leq \|y\|.$$
Hence
\begin{eqnarray}
&&\frac{1}{2}\int_0^T \frac{\sum_{i,j=1}^d w_iw_j \rho_{ij}\,\delta_t^i\delta_t^j\,S^{i}_{t-} S^j_{t-}}{I_{t-}^2}\,dt+ \int_0^T  \int (e^y-1)^2 k(t,\,dy)\,dt\nonumber\\
&=&\frac{1}{2}\int_0^T \frac{\sum_{i,j=1}^d w_iw_j \rho_{ij}\,\delta_t^i\delta_t^j\,S^{i}_{t-} S^j_{t-}}{I_{t-}^2}\,dt\nonumber\\
&+&\int_0^T \int_{\mathbb{R}^d} \left(\frac{\sum_{1\leq i\leq d-1} w_iS_{t-}^i e^{y_i}+w_dS_{t-}^d e^{y}}{I_{t-}}-1\right)^2 \nu(dy_1,\cdots,dy_{d-1},dy)\,dt\nonumber\\
&\leq& \frac{1}{2}\sum_{i,j=1}^d \rho_{ij}\,\delta_t^i\delta_t^j+ \int_0^T \int_{\mathbb{R}^d} (e^{\|y\|}-1)^2 \nu(dy_1,\cdots,dy_{d-1},dy)\,dt.\nonumber
\end{eqnarray}
Using assumptions \eqref{ch3.assum.tris}, the last inequality implies that $I_t$ satisfies \eqref{ch3.H}. Hence Theorem \ref{ch3.pide.forward.prop} can now be applied to $I$, which yields the result.
\end{proof}

\subsection{Forward equations for CDO pricing}\label{ch3.forwardcdo.sec}
Portfolio credit derivatives such as CDOs or index default swaps are derivatives whose payoff depends on the total loss $L_t$ due to defaults in a reference portfolio of obligors.
Reduced-form top-down models of portfolio default risk \cite{filipovic09,giesecke08,schonbucher05,contminca08,spa} represent the default losses of a portfolio as a {\it marked point process} $(L_t)_{t\geq 0}$ where the jump times
represents credit events in the portfolio and the jump sizes $\Delta L_t$ represent the portfolio loss upon a default event.
Marked point processes with random intensities  are increasingly used as ingredients in such models \cite{filipovic09,giesecke08,lopatin08,schonbucher05,spa}.
In all such models the loss process (represented as a fraction of the portfolio notional) may be represented as
$$ L_t = \int_0^t\int_{0}^1 x \,M(ds\, dx),  $$ where
 $M(dt\,dx)$ is an integer-valued random measure with compensator
$$\mu(dt\,dx;\omega)=m(t,dx; \omega)\,dt.$$
If furthermore
\begin{equation}
\int_0^1 x\,m(t,dx)<\infty,
\end{equation}
then  $L_t$ may be expressed in the form
$$ L_t = \int_0^t\int_{0}^1 x \,\left(m(s,dx)\,ds+\tilde{M}(ds\, dx)\right) , $$ where
$$ \int_0^t\int_{0}^1 x \,\tilde{M}(ds\, dx),$$ is a $\mathbb{P}$-martingale.
The point process $ N_t= M([0,t]\times [0,1])$ represents the number of defaults and
$$
 \lambda_t(\omega)=\int_{0}^1 m(t,dx;\omega)
$$
  represents the default intensity. Denote by
$T_1\leq T_2\leq ..$ the jump times of $N$. The cumulative loss process   $L$ may also be represented as
 $$ L_t = \sum_{k=1}^{N_t} Z_k, $$
where the ``mark" $Z_k$, with values in $[0,1]$, is distributed according to
$$
F_t(dx ;\omega)=\frac{m_X(t,dx;\omega)}{\lambda_t(\omega)}.
$$
Note that the percentage loss $L_t$ belongs to $[0,1]$, so  $\Delta L_{t}\in[0,1-L_{t-}]$.
For the equity tranche $[0,K]$, we define the expected tranche notional at maturity $T$ as
\begin{eqnarray}\label{ch3.def.cdo}
C_{t_0}(T,K)=\mathbb{E}[ (K-L_T)_+ |{\cal F}_{t_0}].
\end{eqnarray}
As noted in \cite{contminca08}, the prices of portfolio credit derivatives such as CDO tranches only depend on the loss process through the expected tranche notionals. Therefore, if one is able to compute $C_{t_0}(T,K)$ then one is able to compute the values of all CDO tranches  at date $t_0$.  In the case of a loss process with constant loss increment, Cont and Savescu \cite{contsavescu08} derived a forward equation for the expected tranche notional.
The following result generalizes the forward equation derived by Cont and Savescu \cite{contsavescu08} to a more general setting which allows for random, dependent loss sizes and possible dependence between the loss given default and the default intensity:
\begin{proposition}[Forward equation for expected tranche notionals]\label{ch3.pide.cdo.prop}
Assume there exists a measurable function $m_Y:[0,T]\times[0,1]\mapsto \mathcal{R}([0,1])$ such that for all $t\in[t_0,T]$ and for all $A\in\mathcal{B}([0,1)]$,
\begin{equation}
m_Y(t,A,L_{t-})=E[m_X(t,A,.)|L_{t-}],
\end{equation}
and denote $M_Y(dt\,dy)$ the integer-valued random measure with compensator $ m_Y(t,dy,z)\,dt$. Define the effective default intensity %$(\lambda_t^Y)$% and the mark distribution $F^Y(t,Y_{t-},dy)$:
\begin{equation}
\lambda^Y(t,z)= \int_0^{1-z} m_Y(t,dy,z).
\end{equation}
Then the expected tranche notional $(T,K)\mapsto C_{t_0}(T,K)$, as a function of maturity and
strike, is a solution  of the partial
integro-differential equation
\begin{equation}\label{ch3.forward.pide.cdo}
\begin{split}
&\frac{\partial C_{t_0}}{\partial T}(T,K)\\
&\quad=-\int_{0}^K \frac{\partial^2 C_{t_0}}{\partial K^2}(T,dy)\,\left[\int_{0}^{K-y}(K-y-z)\,m_Y(T,dz,y)-(K-y)\lambda^Y(T,y)\right],
\end{split}
\end{equation}
on $[t_0,\infty[\times ]0,1[$ with
the initial condition:\ $\forall K\in [0,1],\quad C_{t_0}(t_0,K)= (K-L_{t_0})_+.$
\end{proposition}
\begin{proof}
By replacing $\mathbb{P}$ by the conditional measure ${\mathbb{P}}_{|\mathcal{F}_0}$  given $\mathcal{F}_0$, we may
replace the conditional expectation in (\ref{ch3.def.cdo}) by an expectation with
respect to the marginal distribution $p_T(dy)$ of $L_T$ under ${\mathbb{P}}_{|\mathcal{F}_{t_0}}$.
Thus, without loss of generality, we put $t_0=0$ in the sequel and consider the case where ${\mathcal{F}_0}$ is the $\sigma$-algebra generated by all $\mathbb{P}$-null sets.
(\ref{ch3.def.cdo}) can be expressed as
\begin{equation}\label{ch3.def.cdo.bis}
C(T,K)=\int_{\mathbb{R}^+} \left(K-y\right)^+\,p_{T}(dy).
\end{equation}
Differentiating with respect to $K$, we get
\begin{equation}
   \frac{\partial C}{\partial K}=\int_0^K p_{T}(dy)=\mathbb{E}\left[1_{ \{L_{t-}\leq K\}}\right],\quad\quad \frac{\partial^2 C}{\partial K^2}(T,dy)=p_{T}(dy).
\end{equation}
%Let $\mathcal{L}^K_t=\mathcal{L}^K_t(L)$ be the semimartingale local time of $L$ at $K$ under $\mathbb{P}$.
For  $h>0$ applying the Tanaka-Meyer
formula  to $(K-{L}_t)^+$ between $T$ and $T+h$, we have
\begin{equation}
\begin{split}
  (K-L_{T+h})^+&=(K-L_T)^+ -\int_T^{T+h} 1_{ \{L_{t-}\leq K\}} dL_t \\%+ \frac{1}{2} (L^K_{T+h}-L^K_T) \\
  &+ \sum_{T< t\leq T+h} \left[(K-L_{t})^+-(K-L_{t-})^++1_{ \{L_{t-}\leq K\}}\Delta L_{t}\right].
\end{split}
\end{equation}
%%%%%%%%%%%%%%%%%%%%%%%%
%\begin{equation*}
%dL_t=\int_0^1 \left(x\,m(t,dx)\,dt+ x \tilde{M}(dt\, dx)\right)
%\end{equation*}
%\begin{equation*}
%\Delta L_t=L_{t}-L_{t-}=x
%\end{equation*}
%%%%%%%%%%%%%%%%%%%%%%%%%

%%justifier un peu condition d'intégrabilité...
Taking expectations, we get
 \begin{eqnarray*}
    C(T+h,K)-C(T,K)&=&\mathbb{E}\left[\int_T^{T+h} dt\,1_{\{ L_{t-}\leq K\}}\,\int_0^{1-L_{t-}} x\,m(t,dx)\right]  \nonumber\\
    &+&   \mathbb{E}\left[ \sum_{T< t\leq T+h} (K-L_t)^+-(K-L_{t-})^+
+1_{\{ L_{t-}\leq  K\}}\Delta L_{t}\right].\nonumber
  \end{eqnarray*}
%Noting that $S_{t-}
%1_{ \{S_{t-}> K\}}= (S_{t-}-K)^+ + K1_{ \{S_{t-}> K\}} $, we obtain
The first term may be computed as
\begin{eqnarray*}
\mathbb{E}\left[\int_T^{T+h} dt\,1_{\{ L_{t-}\leq K\}}\,\int_0^{1-L_{t-}} x\,m(t,dx)\right]
&=&\int_T^{T+h} dt\,\mathbb{E}\left[1_{\{ L_{t-}\leq K\}}\,\int_0^{1-L_{t-}} x\,m(t,dx)\right]\\
&=&\int_T^{T+h} dt\,\mathbb{E}\left[\mathbb{E}\left[1_{\{ L_{t-}\leq K\}}\,\int_0^{1-L_{t-}} x\,m(t,dx)\Big|L_{t-}\right]\right]\\
&=&\int_T^{T+h} dt\,\mathbb{E}\left[1_{\{ L_{t-}\leq K\}}\int_0^{1-L_{t-}} x\,m_Y(t,dx,L_{t-})\right]\\
&=&\int_T^{T+h} dt\,\int_0^K p_T(dy)\left(\int_0^{1-y} x\,m_Y(t,dx,y)\right).
\end{eqnarray*}
As for the jump term,
\begin{eqnarray*}
  &&\mathbb{E}\left[\sum_{T< t\leq T+h} (K-L_{t})^+-(K-L_{t-})^++1_{ \{L_{t-}\leq K\}} \Delta L_{t}\right] \nonumber\\
  &=& \mathbb{E}\left[\int_T^{T+h}dt\int_0^{1-L_{t-}} m(t,dx)\,\left((K-L_{t-}-x)^+-(K-L_{t-})^++1_{\{ L_{t-}\leq K\}}x\right)\right]\nonumber\\
  &=&\int_T^{T+h}dt\,\mathbb{E}\left[\int_0^{1-L_{t-}} \,m(t,dx)\left((K-L_{t-}-x)^+-(K-L_{t-})^++1_{\{ L_{t-}\leq K\}}x\right)\right]\nonumber\\
&=&\int_T^{T+h}dt\,\mathbb{E}\left[\mathbb{E}\left[\int_0^{1-L_{t-}} \,m(t,dx)\left((K-L_{t-}-x)^+-(K-L_{t-})^++1_{\{ L_{t-}\leq K\}}x\right)\Big|L_{t-}\right]\right]\\
%&=&\int_T^{T+h}dt\mathbb{E}\left[\int_0^1 \,\mathbb{E}\left[m(t,dx)|L_{t-}\right]\left(K-L_{t-}-x)^+-(K-L_{t-})^++1_{\{ L_{t-}\leq K\}}x\right)\right]\\
&=&\int_T^{T+h}dt\,\mathbb{E}\left[\int_0^{1-L_{t-}} \,m_Y(t,dx,L_{t-})\left((K-L_{t-}-x)^+-(K-L_{t-})^++1_{\{ L_{t-}\leq K\}}x\right)\right]\\
&=&\int_T^{T+h} dt\,\int_0^1 p_T(dy)\int_0^{1-y} \,m_Y(t,dx,y)\left((K-y-x)^+-(K-y)^++1_{\{y\leq K\}}x\right),
\end{eqnarray*}
where the inner integrals may be computed as
\begin{eqnarray*}
&&\int_0^1 p_T(dy)\int_0^{1-y} \,m_Y(t,dx,y)\left((K-y-x)^+-(K-y)^++1_{\{y\leq K\}}x\right)\\
&=&\int_0^K p_T(dy)\int_0^{1-y} \,m_Y(t,dx,y)\left((K-y-x)1_{\{ K-y>x\}}-(K-y-x)\right)\\
&=&\int_0^K p_T(dy)  \int_{K-y}^{1-y}  \,m_Y(t,dx,y)(K-y-x).
\end{eqnarray*}
Gathering together all the terms, we obtain
\begin{eqnarray*}
&&C(T+h,K)-C(T,K)\\
&=&\int_T^{T+h} dt\,\int_0^K p_T(dy)\left(\int_0^{1-y} x\,m_Y(t,dx,y)\right)
+\int_T^{T+h} dt\,\int_0^K p_T(dy)  \left(\int_{K-y}^{1-y} \,m_Y(t,dx,y)(K-y-x)\right)\\
&=&\int_T^{T+h} dt\,\int_0^K p_T(dy)\left(-\int_{0}^{K-y}  \,m_Y(t,dx,y)(K-y-x)+(K-y)\lambda^Y(T,y)\right).
%&=&\int_T^{T+h} dt\,\int_0^K p_T(dy)\left(\int_0^1 x\,m_Y(t,dx,y)+\int_0^1 \,m_Y(t,dx,y)\left(K-y-x)^+-(K-y)^++1_{\{y\leq K\}}x\right)
\end{eqnarray*}
Dividing by $h$ and taking the limit $h\to 0$ yields
\begin{eqnarray*}
\frac{\partial C}{\partial T}&=&-\int_{0}^K p_{T}(dy)\,\left[\int_{0}^{K-y}(K-y-x)\,m_Y(T,dx,y)-(K-y)\lambda^Y(T,y)\right]\\
&=&-\int_{0}^K\frac{\partial^2 C}{\partial K^2}(T,dy)\,\left[\int_{0}^{K-y}(K-y-x)\,m_Y(T,dx,y)-(K-y)\lambda^Y(T,y)\right].
\end{eqnarray*}
\end{proof}

In \cite{contsavescu08}, loss given default (i.e. the jump size of $L$) is assumed constant $\delta=(1-R)/n$: then $Z_k=\delta$, so $L_t=\delta N_t$ and one can compute $C(T,K)$  using the law of $N_t$. Setting $t_0=0$ and  assuming as above that $\mathcal{F}_{t_0}$ is generated by null sets, we have
\begin{eqnarray}\label{ch3.def.cdo.savescu}
C(T,K)=\mathbb{E}[(K-L_T)^+]=\mathbb{E}[(k\,\delta-L_T)^+]
=\delta\,\mathbb{E}[(k-N_T)^+]\equiv \delta \,C_k(T).
\end{eqnarray}
The compensator of $L_t$ is $\lambda_t\,\epsilon_{\delta}(dz)\,dt$, where $\epsilon_{\delta}(dz)$ is the point mass at the point $\delta$. The effective compensator becomes
$$m_Y(t,dz,y)=E[\lambda_t|L_{t-}=y]\,\epsilon_{\delta}(dz)\,dt=\lambda^Y(t,y)\,\epsilon_{\delta}(dz),$$
and the effective default intensity is
$\lambda^Y(t,y)= E[\lambda_t|L_{t-}=y].$
Using the notations in \cite{contsavescu08}, if we set $y=j\delta$ then
$$\lambda^Y(t,j\delta)= E[\lambda_t|L_{t-}=j\delta]=E[\lambda_t|N_{t-}=j]=a_j(t)$$ and $p_t(dy)=\sum_{j=0}^n q_j(t)\epsilon_{j\delta}(dy)$.
Let us focus on (\ref{ch3.forward.pide.cdo}) in this case. We recall from the proof of Proposition \ref{ch3.pide.cdo.prop} that
\begin{eqnarray*}
\frac{\partial C}{\partial T}(T,k\delta)&=&\int_{0}^1 p_{T}(dy)\,H_T.\left(k\delta-y\right)^+\nonumber\\
&=&\int_{0}^1 p_{T}(dy)\,\int_{0}^{1-y}[(k\delta-y-z)^+-(k\delta-y)^+]\,\lambda^Y(T,y)\,\epsilon_{\delta}(dz)\nonumber\\
&=&\int_{0}^1 p_{T}(dy)\,\lambda^Y(T,y)\,[(k\delta-y-\delta)^+-(k\delta-y)^+]\,1_{ \{\delta<1-y\}}\nonumber\\
&=&-\delta\sum_{j=0}^n q_j(T)\,a_j(T)\,1_{\{ j\leq k-1\}}.
\end{eqnarray*}
 This expression can be simplified as in \cite[Proposition 2]{contsavescu08}, leading
 to the forward equation
\begin{eqnarray*}
 \frac{\partial C_k(T)}{\partial T}
&=& a_k(T)C_{k-1}(T)-a_{k-1}(T)C_{k}(T)-\sum_{j=1}^{k-2} C_{j}(T)[a_{j+1}(T)-2a_j(T)+a_{j-1}(T)]\nonumber\\
&=&[a_k(T)-a_{k-1}(T)]C_{k-1}(T)-\sum_{j=1}^{k-2}(\nabla^2 a)_jC_{j}(T)-a_{k-1}(T)[C_{k}(T)-C_{k-1}(T)].\nonumber
\end{eqnarray*}
Hence we recover \cite[Proposition 2]{contsavescu08} as a special case of Proposition \ref{ch3.pide.cdo.prop}.

\def\polhk#1{\setbox0=\hbox{#1}{\ooalign{\hidewidth
  \lower1.5ex\hbox{`}\hidewidth\crcr\unhbox0}}}

\end{document}